\documentclass[a4paper,UKenglish,cleveref, autoref]{lipics-v2019}

\hideLIPIcs

\usepackage[boxed,linesnumbered]{algorithm2e}
\usepackage{graphicx}
\let\oldnl\nl
\newcommand{\nonl}{\renewcommand{\nl}{\let\nl\oldnl}}
\newcommand{\Otilde}{\widetilde{O}}

\newcommand{\CPref}{\mathrm{CommonPref}}
\newcommand{\CSuff}{\mathrm{CommonSuff}}
\newcommand{\Detour}{\mathrm{Detour}}

\def\sarel#1{}
\def\shiri#1{}

\makeatletter
\newcommand{\removelatexerror}{\let\@latex@error\@gobble}
\makeatother

\title{Deterministic Combinatorial Replacement Paths and Distance Sensitivity Oracles} 

\author{Noga Alon}
{Department of Mathematics, Princeton University, Princeton, NJ 08544, USA and Schools of Mathematics and Computer Science, Tel Aviv University, Tel Aviv 69978, Israel.}
{nogaa@tau.ac.il}
{}
{Research supported in part by
NSF grant DMS-1855464, ISF grant 281/17 and GIF grant
G-1347-304.6/2016.}

\author{Shiri Chechik}
{Blavatnik School of Computer Science, Tel Aviv University, Tel Aviv 69978, Israel.}
{shiri.chechik@gmail.com}
{}
{Research supported in part by the Israel Science Foundation grant No. 1528/15 and the Blavatnik Fund.}

\author{Sarel Cohen}
{Blavatnik School of Computer Science, Tel Aviv University, Tel Aviv 69978, Israel.}
{sarelcoh@post.tau.ac.il}
{}
{Research supported in part by the Israel Science Foundation grant No. 1528/15 and the Blavatnik Fund.}

\authorrunning{N. Alon, S. Chechik and S. Cohen}

\Copyright{Noga Alon, Shiri Chechik and Sarel Cohen}



\nolinenumbers 
\EventEditors{Christel Baier, Ioannis Chatzigiannakis, Paola Flocchini, and Stefano Leonardi}
\EventNoEds{4}
\EventLongTitle{46th International Colloquium on Automata, Languages, and Programming (ICALP 2019)}
\EventShortTitle{ICALP 2019}
\EventAcronym{ICALP}
\EventYear{2019}
\EventDate{July 9--12, 2019}
\EventLocation{Patras, Greece}
\EventLogo{eatcs}
\SeriesVolume{132}
\ArticleNo{7}

\begin{document}

\maketitle

\begin{abstract}
In this work we derandomize two central results in graph algorithms,
replacement paths and distance sensitivity oracles (DSOs) matching
in both cases the running time of the randomized algorithms.

For the replacement paths problem, let $G = (V,E)$ be a directed
unweighted graph with $n$ vertices and $m$ edges and let $P$ be a
shortest path  from $s$ to $t$ in $G$. The {\sl replacement paths}
problem is to find for every edge $e \in P$ the shortest path from
$s$ to $t$ avoiding $e$. Roditty and Zwick [ICALP 2005] obtained
 a randomized algorithm with running time of $\Otilde(m
\sqrt{n})$. Here we provide the first deterministic algorithm for
this problem, with the same $\Otilde(m \sqrt{n})$ time. Due to
matching conditional lower bounds of Williams {\sl et. al.} [FOCS
2010], our deterministic combinatorial algorithm for the replacement
paths problem is optimal up to polylogarithmic factors (unless the
long standing bound of $\Otilde(mn)$ for the combinatorial boolean
matrix multiplication can be improved). This also implies a
deterministic algorithm for the second simple shortest path problem
in $\Otilde(m \sqrt{n})$ time, and a deterministic algorithm for the
$k$-simple shortest paths problem in $\Otilde(k m \sqrt{n})$ time
(for any integer constant $k > 0$).


For the problem of distance sensitivity oracles, let $G = (V,E)$ be
a directed graph with real-edge weights. An $f$-Sensitivity Distance
Oracle ($f$-DSO) gets as input the graph $G=(V,E)$ and a parameter
$f$, preprocesses it into a data-structure, such that given a query
$(s,t,F)$ with $s,t \in V$ and $F \subseteq E \cup V, |F| \le f$
being a set of at most $f$ edges or vertices (failures), the query
algorithm efficiently computes the distance from $s$ to $t$ in the
graph $G \setminus F$ ({\sl i.e.}, the distance from $s$ to $t$ in
the graph $G$ after removing from it the failing edges and vertices
$F$).

For weighted graphs with real edge weights, Weimann and Yuster [FOCS 2010] presented several randomized $f$-DSOs.
In particular, they presented a combinatorial $f$-DSO with $\Otilde(mn^{4-\alpha})$ preprocessing time and subquadratic $\Otilde(n^{2-2(1-\alpha)/f})$ query time, giving a tradeoff between preprocessing and query time for every value of $0 < \alpha < 1$.
We derandomize this result and present a combinatorial deterministic $f$-DSO with the same asymptotic preprocessing and query time.
\end{abstract}

\section{Introduction}


In many algorithms used in computing environments such as massive
storage devices, large scale parallel computation, and communication
networks, recovering from failures must be an integral part.
Therefore, designing algorithms and data structures whose running
time is efficient even in the presence of failures is an important
task. In this paper we study variants of shortest path queries in
setting with failures.

The computation of shortest paths and distances in the presence of failures was extensively studied.
Two central problems researched in this field are the Replacement Paths problem and Distance Sensitivity Oracles, we define these problems hereinafter.

{\bf The Replacement Paths problem} (See, {\sl e.g.},
\cite{Roditty2005, WilliamsRP11, GoLe09,EmPeRo10, Klein10, WY13,
Bernstein10, Yen71, Law72, LeeL14, MaMiGu89, NaPrWi01, WilliamsW10,
Epp98}).
  Let $G=(V,E)$ be a graph (directed or undirected, weighted or unweighted) with $n$ vertices and
$m$ edges and let $P_G(s,t)$ be a shortest path from $s$ to $t$.
For every edge $e \in P_G(s,t)$ a replacement path $P_G(s,t,e)$ is a shortest path from $s$ to $t$
in the graph $G \setminus \{e\}$ (which is the graph $G$ after removing the edge $e$).
Let $d_G(s,t,e)$ be the length of the path $P_G(s,t,e)$.
The replacement paths
problem is as follows: given a shortest path $P_G(s,t)$ from $s$ to $t$ in $G$,
compute $d_G(s,t,e)$ (or an approximation of it) for every $e \in P_G(s,t)$.

{\bf Distance Sensitivity Oracles} (See, {\sl e.g.}, \cite{ChCoFiKa17, GW12, BeKa08, BK09, CLPR10, DT02, DeThChRa08, DP09, KB10}).
  An $f$-Sensitivity Distance Oracle ($f$-DSO) gets as input a graph $G=(V,E)$ and a
parameter $f$, preprocesses it into a data-structure, such that given
a query $(s,t,F)$ with $s,t \in V$ and $F \subseteq E \cup V, |F| \le f$ being a set of at most $f$
edges or vertices (failures), the query algorithm efficiently computes (exactly or approximately) $d_G(s,t,F)$ which is the distance from $s$ to $t$
in the graph $G \setminus F$ ({\sl i.e.}, in the graph $G$ after removing from it the failing edges and vertices $F$).
Here we would like to optimize several parameters of the data-structure: minimize the size of the oracle, support many failures $f$, have efficient preprocessing and query algorithms, and if the output is an approximation of the distance then optimize the approximation-ratio.

An important line of research in the theory of computer science is derandomization.
In many algorithms and data-structures there exists a gap between the best known randomized algorithms and the best known deterministic algorithms.
There has been extensive research on closing the gaps between the best known randomized and deterministic algorithms in many problems or proving that no deterministic algorithm can perform as good as its randomized counterpart. There also has been a long line of work on developing derandomization techniques, in order to obtain deterministic versions of randomized algorithms ({\sl e.g.}, Chapter 16 in \cite{alon2011probabilistic}).

In this paper we derandomize algorithms and data-structures for computing distances and shortest paths in the presence of failures.
Many randomized algorithms for computing shortest paths and distances use variants of the following sampling lemma (see Lemma 1 in Roditty and Zwick \cite{Roditty2005}).

\shiri{I think that at this point you want to mention that the following two lemmas are not the main result of the paper and also mention in short  what is the main contribution - I added a stentense pls take a look}
\sarel{Ok, I think the sentence you added is good. Later there is a complete Section \ref{sec:framework} where we address it in more details.}

\begin{lemma} [Lemma 1 in \cite{Roditty2005}] \label{lem:sampling-roditty}
Let $D_1, D_2, \ldots, D_q \subseteq V$ satisfy $|D_i| > L$ for $1 \le i \le q$ and $|V|=n$. If $R \subseteq V$ is a random subset obtained by selecting each vertex, independently, with probability $(c \ln n)/L$, for some $c>0$, then with probability of at least $1 - q \cdot n^{-c}$ we have $D_i \cap R \ne \emptyset$ for every $1 \le i \le q$.
\end{lemma}

Our derandomization step of Lemma \ref{lem:sampling-roditty} is very simple, as described in Section \ref{sec:framework},
we use the folklore greedy approach to prove the following lemma, which is a deterministic version of Lemma \ref{lem:sampling-roditty}.

\begin{lemma} \label{lemma:greedy} [See also Section \ref{sec:framework}]
Let $D_1, D_2, \ldots, D_q \subseteq V$ satisfy $|D_i| > L$ for $1 \le i \le q$ and $|V|=n$.
One can deterministically find in $\Otilde(qL)$ time a set $R \subset V$
such that $|R| = \Otilde(n/L)$ and $D_i \cap R \ne \emptyset$ for every $1 \le i \le q$.
\end{lemma}

We emphasize that the use of Lemma \ref{lemma:greedy} is very standard and is not our main contribution.
The main technical challenge is how to efficiently and deterministically compute a small number of sets
$D_1, D_2, \ldots, D_q \subseteq V$ so that the invocation of Lemma \ref{lemma:greedy} is fast.

\subsection{Derandomizing the Replacment Paths Algorithm of Roditty
and Zwick \cite{Roditty2005}}

We derandomize the algorithm of Roditty and Zwick \cite{Roditty2005}
and obtain a near optimal deterministic algorithm for the
replacement paths problem in directed unweighed graphs (a problem
which was open for more than a decade since the randomized algorithm
was published) as stated in the following theorem.

\begin{theorem} \label{thm:replacement}
There exists a deterministic algorithm for the replacement paths problem in unweighted directed graphs whose runtime is $\Otilde(m\sqrt{n})$.
This algorithm is near optimal assuming the conditional lower bound of combinatorial boolean matrix multiplication of \cite{WilliamsW10}.
\end{theorem}

The term ``combinatorial
algorithms'' is not well-defined, and it is often interpreted as
non-Strassen-like algorithms \cite{BaDeHoSc12}, or more intuitively,
algorithms that do not use any matrix multiplication tricks.
Arguably, in practice, combinatorial algorithms are to some extent
considered more efficient since the constants hidden in the matrix
multiplication bounds are high. On the other hand, there has been
research done to make fast matrix multiplication practical, {\sl
e.g.}, \cite{HuRiMaGe17, AuGr15}.

Vassilevska Williams and Williams \cite{WilliamsW10} proved a
subcubic equivalence between $\sqrt{n}$ occurrences of the combinatorial replacement paths
problem in unweighted directed graphs and the combinatorial boolean
multiplication (BMM) problem.
More precisely, they proved that there exists some
fixed $\epsilon >0$ such that the combinatorial replacement paths problem
can be solved in $O(mn^{1/2-\epsilon})$ time if and only if there exists some
fixed $\delta > 0$ such that the combinatorial boolean matrix
multiplication (BMM) can be solved in subcubic $O(n^{3-\delta})$
time. Giving a subcubic combinatorial algorithm to the BMM problem,
or proving that no such algorithm exists, is a long standing open
problem. This implies that either both problems can be polynomially
improved, or neither of them does. Hence, assuming the conditional
lower bound of combinatorial BMM, our combinatorial $\Otilde(m \sqrt{n})$ algorithm for the
replacement paths problem in unweighted directed graphs is essentially optimal (up to $n^{o(1)}$ factors).

The replacement paths problem is related to the $k$ simple shortest
paths problem, where the goal is to find the $k$ simple shortest
paths between two vertices. Using known reductions from the
replacement paths problem to the $k$ simple shortest paths problem,
we close this gap as the following Corollary states.

\begin{corollary}
There exists a deterministic algorithm for computing $k$ simple shortest paths in unweighted directed graphs whose runtime is $\Otilde(k m\sqrt{n})$.
\end{corollary}

More related work can be found in Section \ref{appendix:more-related-work}.
As written in Section \ref{appendix:more-related-work}, the trivial $\Otilde(mn)$ time algorithm for solving the replacement paths problem in directed weighted graphs (simply, for every edge $e \in P_G(s,t)$ run Dijkstra in the graph $G\setminus \{e\}$) is deterministic and near optimal (according to a conditional lower bound by \cite{WilliamsW10}).
To the best of our knowledge the only deterministic combinatorial
algorithms known for directed unweighted graphs are the algorithms
for general directed weighted graphs whose runtime is $\Otilde(mn)$
leaving a significant gap between the randomized and deterministic
algorithms. As mentioned above, in this paper we derandomize the
$\Otilde(m \sqrt{n})$ algorithm of Roditty and Zwick
\cite{Roditty2005} and close this gap.


%
%

\subsection{Derandomizing the Combinatorial Distance Sensitivity Oracle of
Weimann and Yuster \cite{WY13}}

Our second result is derandomizing the combinatorial distance
sensitivity oracle of Weimann and Yuster \cite{WY13} and obtaining
the following theorem.

\begin{theorem} \label{thm:dso}
Let $G=(V,E)$ be a directed graph with real edge weights, let
$|V|=n$ and $|E|=m$. There exists a deterministic algorithm that
given $G$ and parameters $f = O(\frac{\log n}{\log \log n})$ and $0
< \alpha < 1$ constructs an $f$-sensitivity distance oracle in
$\Otilde(mn^{4-\alpha})$ time. Given a query $(s,t,F)$ with $s,t \in
V$ and $F \subseteq E \cup V, |F| \le f$ being a set of at most $f$
edges or vertices (failures), the deterministic query algorithm
computes in $\Otilde(n^{2-2(1-\alpha)/f})$  time the distance from
$s$ to $t$ in the graph $G \setminus F$.
\end{theorem}

We remark that while our focus in this paper is in computing
distances, one may obtain the actual shortest path in time
proportional to the number of edges of the shortest paths, using the
same algorithm for obtaining the shortest paths in the replacement
paths problem \cite{Roditty2005}, and in the distance sensitivity
oracles case \cite{WY13}.

\subsection{Technical Contribution and Our Derandomization Framework} \label{sec:framework}
Let ${\cal A}$ be a random algorithm that uses Lemma
\ref{lem:sampling-roditty} for sampling a subset of vertices $R
\subseteq V$. We say that ${\cal P} = \{D_1, \ldots, D_q\}$ is a set
of {\sl critical} paths for the randomized algorithm ${\cal A}$ if
${\cal A}$ uses the sampling Lemma \ref{lem:sampling-roditty} and it
is sufficient for the correctness of algorithm ${\cal A}$ that $R$
is a hitting set for  ${\cal P}$ ({\sl i.e.}, every path in ${\cal
P}$ contains at least one vertex of $R$). According to Lemma
\ref{lemma:greedy} one can derandomize the random selection of the
hitting set $R$ in time that depends on the number of paths in
${\cal P}$. Therefore, in order to obtain an efficient
derandomization procedure, we want to find a small set ${ \cal P}$
of critical paths for the randomized algorithms.

Our main technical contribution is to show how to compute a small set of critical paths that is sufficient to be used as input for the greedy algorithm stated in Lemma \ref{lemma:greedy}.

Our framework for derandomizing algorithms and data-structures that use the sampling Lemma \ref{lem:sampling-roditty} is given in Figure \ref{fig:framework}.

\begin{figure} [h]
\begin{algorithm}[H]
\label{fig:framework}
    \DontPrintSemicolon
    {\bf Step 1}: Prove the existence of a small set of critical paths $\{D_1, \ldots, D_q \}$ such that $|D_i| > L$ and show that it is sufficient for the correctness of the randomized algorithm that the set $R$ obtained by Lemma \ref{lem:sampling-roditty} hits all the paths $D_1, \ldots, D_q $. \;
    {\bf Step 2}: Find an efficient algorithm to compute the paths $D_1, \ldots, D_q$. \;
    {\bf Step 3}: Use a deterministic algorithm to compute a small subset $R \subseteq V$ of vertices such that $D_i \cap R \ne \emptyset$ for every $1 \le i \le q$. For example, one can use the greedy algorithm of Lemma \ref{lemma:greedy} or the blocker set algorithm of \cite{King99} to find a subset $R \subset V$ of $\Otilde(n/L)$ vertices. \;
\end{algorithm}
\caption{Our derandomization framework to derandomize algorithms that use the sampling Lemma \ref{lem:sampling-roditty}.}
\end{figure}


Our first main technical contribution, denoted as Step 1 in Figure \ref{fig:framework}, is proving the existence of small sets of critical paths for the randomized replacement path algorithm of Roditty and Zwick \cite{Roditty2005}
and for the distance sensitivity oracles of Weimann and Yuster \cite{WY13}.
Our second main technical contribution, denoted as Step 2 in Figure \ref{fig:framework}, is developing algorithms to efficiently compute these small sets of critical paths.

For the replacement paths problem, Roditty and Zwick
\cite{Roditty2005} proved the existence of a critical set of
$O(n^2)$ paths, each path containing at least $\lceil \sqrt{n}
\rceil$ edges. Simply applying Lemma \ref{lemma:greedy} on this set
of paths requires $\Otilde(n^{2.5})$ time which is too much, and it
is also not clear from their algorithm how to efficiently compute
this set of critical paths. As for Step 1, we prove the existence of
a small set of $O(n)$ critical paths, each path contains $\lceil
\sqrt{n} \rceil$ edges, and for Step 2, we develop an efficient
algorithm that computes this set of critical paths in
$\Otilde(m\sqrt{n})$ time.

For the problem of distance sensitivity oracles, Weimann and Yuster
\cite{WY13} proved the existence of a critical set of $O(n^{2f+3})$
paths, each path containing $n^{(1-\alpha)/f}$ edges (where $0 <
\alpha < 1$). Simply applying Lemma \ref{lemma:greedy} on this set
of paths requires $\Otilde(n^{2f+3 + (1-\alpha)/f})$ time which is
too much, and here too, it is also not clear from their algorithm
how to efficiently and deterministically compute this set of
critical paths. As for Step 1, we prove the existence of a small set
of $O(n^{2+\epsilon})$ critical paths, each path contains
$n^{(1-\alpha)/f}$ edges, and for Step 2, we develop an efficient
deterministic algorithm that computes this set of critical paths in
$\Otilde(mn^{1+\epsilon})$ time.

For Step 3, we use the folklore greedy deterministic algorithm denoted here by
\newline GreedyPivotsSelection$(\{ D_1, \ldots, D_q \})$.
Given as input the paths $D_1, \ldots, D_q$,
each path contains at least $L$ vertices, the algorithm chooses a set of pivots $R \subseteq
V$ such that for every $1 \le i \le q$ it holds that $D_i \cap
R \ne \emptyset$. In addition, it holds that $|R| = \Otilde(\frac{n}{L})$ and the runtime of the algorithm is $\Otilde(qL)$.

The GreedyPivotsSelection algorithm works as follows.
Let $\mathcal{P} = \{ D_1, \ldots, D_q \}$.
Starting with $R \gets \emptyset$, find a vertex $v \in V$ which is contained in the
maximum number of sets of $\mathcal{P}$, add it to $R$
and remove all the sets that contain $v$ from $\mathcal{P}$. Repeat
this process until $\mathcal{P} = \emptyset$.


\begin{lemma} \label{lemma:greedy-correctness}
Let $1 \le L \le n$ and $1 \le q < poly(n)$ be two integers.
Let $D_1, \ldots, D_q \subseteq V$ be paths satisfying $|D_i| \ge L$ for every $1 \le i \le q$.
The algorithm GreedyPivotsSelection$( \{D_1, \ldots, D_q \})$ finds in $\Otilde(qL)$ time a set $R \subset V$
such that for every $1 \le i \le q$ it holds that $R \cap D_i \ne
\emptyset$ and $|R| = O(\frac{n \log q}{L}) = \Otilde(n/L)$.
\end{lemma}


\begin{proof}
We first prove that for every $1 \le i \le q$ it holds that $R \cap D_i \ne
\emptyset$ and $|R| = O(\frac{n \log q}{L}) = \Otilde(n/L)$.

When the algorithm terminates then every set $D \in
\mathcal{D}$ contains at least one of the vertices of $R$, as otherwise
$\mathcal{D}$ would have contained the sets which are disjoint from
$R$ and the algorithm should have not finished since $\mathcal{D}
\ne \emptyset$.

For every vertex $v \in V$, let $c(v)$ be a variable which denotes, at every moment of the algorithm, the number of sets in $\mathcal{D}$ which contain $v$.

Denote by $\mathcal{D}_i$ the set $\mathcal{D}$ after $i$ iterations.
Let $\mathcal{D}_0 = \{ D_1, \ldots, D_q\}$ be the initial set $\mathcal{D}$ given as input to
the algorithm, then $|\mathcal{D}_0| = q$.  We claim that the process
terminates after at most $\Otilde(n/L)$ iterations, and since at
every iteration we add one vertex $v$ to $R$, it follows that $|R| =
\Otilde(n/L)$.  Recall that $\mathcal{D}$ contains sets of size
at least $L$.  Hence, $\Sigma_{v \in V} c(v) \ge |\mathcal{D}| L$.
It follows that the average number of sets that a vertex $v \in V$
belongs to is:
$avg = \frac{\Sigma_{v \in V} c(v)}{n} \ge \frac{|\mathcal{D}| L }{n}$.
By the pigeonhole principle, the vertex $v_i = \arg \max_{v \in V}
\{ c(v) \}$  belongs to at least $\frac{|\mathcal{D}|L}{n}$ sets of
$\mathcal{D}$. Therefore, $|\mathcal{D}_i| = |\{ D \in \mathcal{D}
| v_i \in D \}| \ge \frac{|\mathcal{D}|L}{n}$. At iteration $i$
we remove from $\mathcal{D}$ the sets $\mathcal{D}_i$, so in each
iteration we decrease the size of $\mathcal{D}$ by at least a factor of
$(1-L/n)$. After the $i^{\text{th}}$ iteration, the size of $\mathcal{D}$
is at most $(1-L/n)^i |\mathcal{D}_0|$. Therefore, after the $i =
(n/L)\ln q + 1$ iteration, the size of $\mathcal{D}$ is at most $(1-L/n)^i
|\mathcal{D}_0| < 1/q |\mathcal{D}_0| \le 1$, where the last inequality
holds since $|\mathcal{D}_0|=q$.  It follows that
after $(n/L) \ln q + 1$ iterations we have $\mathcal{D} = \emptyset$.

At each iteration we add one vertex $v_i$ to the set $R$, thus the size
of the set $R$ is $\Otilde(n/L)$.

Next we describe an implementation of the GreedyPivotsSelection algorithm (see Figure
\ref{fig:greedy-pivots-selection} for pseudo-code).
The first thing we do is keep only an arbitrary subset of $L$ vertices from every
$D \in \mathcal{D}$ so that every set $D \in \mathcal{D}$ contains exactly
$L$ vertices.

We implement the algorithm GreedyPivotsSelection as follows.
During the runtime of the algorithm we
maintain a counter $c(v)$ for every vertex $v \in V$ which equals the
number of sets in $\mathcal{D}$ that contain $v$.
During the initialization
of the algorithm, we construct a subset of vertices $V' \subseteq V$ which contains all the vertices in all the
paths $\mathcal{D}$, and compute
we compute $c(v)$ directly, first by setting $\forall_{v
\in V'} c(v) \gets 0$ and then we scan all the sets $D \in \mathcal{D}$
and every vertex $v \in D$ and increase the counter $c(v)
\gets c(v) + 1$. After this initialization we have $c(v) = |\{ D \in
\mathcal{D} | v \in D \}|$ which is the number of sets of $\mathcal{D}$
that contain $v$. We further initialize a binary search tree $BST$ and insert
every vertex $v \in V'$ into $BST$ with the key $c(v)$, and initialize $R \gets \emptyset$.
We also create a list $L(v)$ for every vertex $v \in V'$ which contains pointers
to the sets $D \in \mathcal{D}$ that contain $v$. Hence, $L(v) = \{
D \in \mathcal{D} | v \in D \}$ and $c(v) = |L(v)|$.

To obtain the set $R$ we run the following loop. While $\mathcal{D}
\ne \emptyset$ we find the vertex $v \in V'$ which is contained in the
maximum number of paths of $\mathcal{D}$ and add $v$ to $R$. The vertex
$v$ is computed in $O(\log n)$ time by extracting the element in $BST$ whose key is maximal.
Then we remove from $\mathcal{D}$ all the sets which contain $v$ (these are exactly the sets
$L(v)$) and we update the counters $c(v)$ by scanning every set $D \in
L(v)$ and every vertex $u \in D$ and decreasing the counter $c(u)$ by one (we also update the key of $u$ in $BST$ to the counter $c(u)$).

We analyse the runtime of this greedy algorithm. Computing the subset of vertices $V' \subseteq V$
and setting all the values $c(v) \gets 0$ at the beginning for every $v \in V'$  takes
$\Otilde(qL)$ time. Computing the values $c(v) = |\{ D \in \mathcal{D} |
v \in D \}|$ takes $O(qL)$ time as we loop over all the $q$
sets $D \in \mathcal{D}$ and for every $D$ we loop over the exactly $L$
vertices $v \in D$ and increase the counter $c(v)$ by one. Initializing the binary search tree $BST$
and inserting to it every vertex $v \in V'$ with key $c(v)$ takes $\Otilde(|V'|) = \Otilde(qL)$ time, and all the extract-max
operations on $BST$ take additional $O(|V'|) = \Otilde(qL)$ time.
The total time of operations of
the form $c(v) \gets c(v)-1$ is $O(qL)$ as this is the sum of all
values $c(v)$ at the beginning and each such operation is handled in $O(\log n)$ time
by updating the key of the vertex $v$ in $BST$ to $c(v)-1$. The total time for checking the
lists $L(v)$ of all vertices chosen to $R$ is at most $O(qL)$, as
this is the sum of sizes of all sets $L(v)$. Therefore, the total
running time is $\Otilde(qL)$.
\end{proof}

\begin{figure}
\removelatexerror
\begin{algorithm}[H]
\label{fig:greedy-pivots-selection}
    \nonl \TitleOfAlgo{GreedyPivotsSelection$(\{ D_1, \ldots, D_q \})$}
    \DontPrintSemicolon
    \tcc*[l]{Initialization}
    $V' \gets \emptyset$ \;
    \For {$D \in \mathcal{D}$}
    {\For {$v \in D$}
    {$V' \gets V' \cup \{ v \}$ \; }}

    \For {$v \in V'$}
    { $c(v) \gets 0, L(v) \gets \emptyset$ \;
    }

    \For {$D \in \mathcal{D}$}
    {\For {$v \in D$}
    { $c(v) \gets c(v) + 1$ \;
    $L(v)$.append($D$) \tcc*[r]{Append the list $L(v)$ with a pointer to the set $D$}
    }}

    $BST \gets $ Empty-Binary-Search-Tree() \;
    \For {$v \in V'$}
    { Insert $v$ to the binary-search $BST$ with the key $c(v)$. \;
    }

    $R \gets \emptyset$. \;
    \tcc*[l]{Loop Invariant: $c(v) = |\{D \in \mathcal{D} | v \in D \}|$}
    \While{$\mathcal{D} \ne \emptyset$}
    {
        $v = BST.Extract-Max()$ \tcc*[r]{$v$ is the vertex in $BST$ whose key $c(v)$ is maximal.}
        $R \gets R \cup \{ v \}$ \;
        \For {$D \in L(v)$}
        {\If {$D \in \mathcal{D}$}
        {\For {$u \in D$}
        {
        $BST$.remove($u$) \;
        $c(u) \gets c(u) - 1$ \;
        Insert $u$ to the binary-search $BST$ with the key $c(u)$. \;
        }
        $\mathcal{D}$.delete($D$)
        }
    }
    }
\end{algorithm}
\caption{Algorithm GreedyPivotsSelection}
\end{figure}

\subsection{Related Work - the Blocker Set Algorithm of King}
We remark
that the GreedyPivotsSelection algorithm is similar to the blocker
set algorithm described in \cite{King99} for finding a hitting set
for a set of paths. The blocker set algorithm was used in
\cite{King99} to develop sequential dynamic algorithms for the APSP
problem. Additional related work is that of Agarwal {\sl et. al. }
\cite{AgarwalRKP18}. They presented a deterministic distributed
algorithm to compute APSP in an edge-weighted directed or undirected
graph in $\Otilde(n^{3/2})$ rounds in the Congest model by
incorporating a deterministic distributed version of the blocker set
algorithm.

While our derandomization framework uses the greedy algorithm (or
the blocker set algorithm) to find a hitting set of vertices for a
critical set of paths $D_1, \ldots, D_q$, we stress
that our main contribution are the techniques to reduce the number of sets $q$ the greedy algorithm must hit (Step 1), and the algorithms to efficiently compute the sets $D_1, \ldots, D_q$ (Step 2).
 These techniques are our main contribution, which enable us to use the greedy algorithm (or the blocker set algorithm) for a wider range of problems.
Specifically, these techniques allow us to derandomize the best known random algorithms for the replacement paths problem and distance sensitivity oracles.
We believe that our techniques can also be leveraged for additional related problems which use a sampling lemma similar
to Lemma \ref{lem:sampling-roditty}.

\subsection{More Related Work} \label{appendix:more-related-work}
We survey related work for the replacement paths problem and distance sensitivity oracles.

{\bf The replacement paths problem.}
The replacement paths problem is motivated by several different applications and has been extensively studied in the last few decades (see e.g. \cite{MaMiGu89,HeSu01,HeSu02,NaPrWi01,WilliamsW10,GoLe09,Roditty2005,EmPeRo10,Klein10,Bernstein10}).
It is well motivated by its own right from the fault-tolerance perspective.
In many applications it is desired to find algorithms and data-structures that are resilient to failures.
Since links in a network can fail, it is important to find backup shortest paths between important
vertices of the graph.

Furthermore, the replacement paths problem is also motivated by
several applications. First, the fastest algorithms to compute the
$k$ simple shortest paths between $s$ and $t$  in directed graphs
executes $k$ iterations of the replacement
paths between $s$ and $t$ in total $\Otilde(mnk)$ time
(see \cite{Yen71,Law72}). Second, considering path auctions,
suppose we would like to find the shortest path from $s$ to $t$ in a
directed graph $G$, where links are owned by selfish agents. Nisan
and Ronen \cite{NiRo01} showed that Vickrey Pricing is an
incentive compatible mechanism, and in order to compute the Vickery
Pricing of the edges one has to solve the replacement paths problem.
It  was raised as an open problem by Nisan and Ronen \cite{NiRo01}
whether there exists an efficient algorithm for solving the
replacement paths problem. In biological sequence alignment
\cite{ByWa84} replacement paths can be used to compute which pieces
of an alignment are most important.

The replacement paths problem has been studied extensively, and by now near optimal algorithms
are known for many cases of the problem.
For instance, the case of undirected graphs admits deterministic near linear solutions (see \cite{MaMiGu89,HeSu01, HeSu02,NaPrWi01}).
In fact, Lee and Lu present linear $O(n+m)$-time algorithms for the replacement-paths problem in on the following classes of $n$-node $m$-edge graphs: (1) undirected graphs in the word-RAM model of computation, (2) undirected planar graphs, (3) undirected minor-closed graphs, and (4) directed acyclic graphs.

A natural question is whether a near linear time algorithm is also possible for the directed case.
Vassilevska Williams and Williams \cite{WilliamsW10} showed that such an algorithm is essentially not possible by presenting
conditional lower bounds.
More precisely, Vassilevska Williams and Williams \cite{WilliamsW10} showed a subcubic equivalence between the combinatorial all pairs shortest paths (APSP) problem and the combinatorial replacement paths problem.
They proved that there exists a fixed $\epsilon >0$ and an
$O(n^{3-\epsilon})$ time combinatorial algorithm for the replacement paths problem if and only if
there exists a fixed $\delta>0$ and an $O(n^{3-\delta})$ time combinatorial algorithm for the APSP problem.
This implies that either both problems admit truly subcubic algorithms, or neither of them does.
Assuming the conditional lower bound that no subcubic APSP algorithm exists, then the trivial algorithm of computing Dijkstra from $s$ in every graph $G \setminus \{e\}$ for every edge $e \in P_G(s,t)$, which takes $O(mn + n^2\log n)$ time, is essentially near optimal.

The near optimal algorithms for the undirected case and the conditional lower bounds for the directed case seem to close the problem.
However, it turned out that if we consider the directed case with bounded edge weights then the picture is not yet complete.

For instance, if we assume that the graph is directed with integer weights in the
range $[-M, M]$ and allow algebraic solutions (rather than combinatorial ones), then Vassilevska Williams presented \cite{WilliamsRP11} an $\Otilde(Mn^\omega)$ time algebraic randomized algorithm
for the replacement paths problem, where $2 \le \omega < 2.373$ is
the matrix multiplication exponent, whose current best known upper
bound is $2.3728639$ (\cite{LeGall14,Williams12,CoppersmithW90}).



Bernstein presented in \cite{Bernstein10} a $(1+\epsilon)$-approximate deterministic replacement paths algorithm which is near optimal (whose runtime is $\Otilde((m \log(nC /c) / \epsilon)$, where $C$ is the largest edge weight in the graph and $c$ is the smallest edge weight).

For unweighted directed graphs the gap between randomized and deterministic solutions is even larger for sparse graphs.
Roditty and Zwick \cite{Roditty2005} presented a randomized
algorithm whose runtime is $\Otilde(m\sqrt{n})$ time for the
replacement paths problem for unweighted directed graphs.
Vassilevska Williams and Williams \cite{WilliamsW10} proved a
subcubic equivalence between the combinatorial replacement paths
problem in unweighted directed graphs and the combinatorial boolean
multiplication (BMM) problem. They proved that there exists some
fixed $\epsilon >0$ such that the combinatorial replacement paths problem
can be solved in $O(mn^{1/2-\epsilon})$ time if and only if there exists some
fixed $\delta > 0$ such that the combinatorial boolean matrix
multiplication (BMM) can be solved in subcubic $O(n^{3-\delta})$
time. Giving a subcubic combinatorial algorithm to the BMM problem,
or proving that no such algorithm exists, is a long standing open
problem. This implies that either both problems can be polynomially
improved, or neither of them does. Hence, assuming the conditional
lower bound of combinatorial BMM, the randomized algorithm of
Roditty and Zwick \cite{Roditty2005} is near optimal.

In the deterministic regime no algorithm for the directed case is known that is asymptotically better (up to ploylog) than invoking APSP algorithm.
Interestingly, in the fault-tolerant and the dynamic settings many of the existing algorithms are randomized,
and for many of the problems there is a polynomial gap between
the best randomized and deterministic algorithms (see e.g.
sensitive distance oracles \cite{GW12}, dynamic shortest paths \cite{HenzingerKNFOCS14,BernsteinC16}, dynamic strongly connected components \cite{HenzingerKN14,HenzingerKN15,CDILP16}, dynamic matching \cite{solomon2016fully,ArChCoStWa17}, and many more).
Randomization is a powerful tool in the classic setting of graph algorithms with full knowledge and is often used to simplify the algorithm and to speed-up its running time.
However, physical computers are deterministic machines, and obtaining true randomness can be a hard task to achieve.
A central line of research is focused on the derandomization of algorithms that relies on randomness.

Our main contribution is a derandomization of the replacement paths algorithm of \cite{Roditty2005} for the case of unweighted directed graphs.
After more than a decade we give the first deterministic algorithm for the replacement paths problem, whose runtime is $\Otilde(m\sqrt{n})$. Our deterministic algorithm matches the runtime of the randomized algorithm, which is near optimal assuming the conditional lower bound of combinatorial boolean matrix multiplication \cite{WilliamsW10}.
In addition, to the best of our knowledge this is the first deterministic solution for the directed case that is asymptotically better than the APSP bound.

The replacement paths problem is related to the $k$ shortest paths problem, where the goal is to find the $k$ shortest paths between two vertices.
Eppstein \cite{Epp98} solved the $k$ shortest paths problem for directed graphs with nonnegative edge weights in $O(m + n\log n + k)$ time. However, the $k$ shortest paths may not be simple, {\sl i.e.}, contain cycles. The problem of $k$ simple shortest paths (loopless) is more difficult. The deterministic algorithm by Yen \cite{Yen71} (which was generalized by Lawler \cite{Law72}) for finding $k$ simple shortest paths in weighted directed graphs can be implemented in $O(kn(m + n\log n))$ time. This algorithm essentially uses in each iteration a replacement paths algorithm. Roditty and Zwick \cite{Roditty2005} described how to reduce the problem of $k$ simple shortest paths into $k$ executions of the second shortest path problem. For directed unweighted graphs, the randomized replacement paths algorithm of Roditty and Zwick \cite{Roditty2005} implies that the $k$ simple shortest paths has a randomized $\Otilde(k m \sqrt{n})$ time algorithm. To the best of our knowledge no better deterministic algorithm is known than the algorithms for general directed weighted graphs, yielding a significant gap between randomized and the deterministic $k$ simple shortest paths for directed unweighted graphs. Our deterministic replacement paths algorithm closes this gap and gives the first deterministic $k$ simple shortest paths algorithm for directed unweighted graphs whose runtime is $\Otilde(k m \sqrt{n})$.

The best known randomized algorithm for the $k$ simple shortest
paths problem in directed unweighted graphs takes $\Otilde(k
m\sqrt{n})$ time (\cite{Roditty2005}), leaving a significant gap
compared to the best known deterministic algorithm which takes
$\Otilde(k m n)$ time ({\sl e.g.}, \cite{Yen71}, \cite{Law72}). We
close this gap by proving the existence of a deterministic algorithm
for computing $k$ simple shortest paths in unweighted directed
graphs whose runtime is $\Otilde(k m\sqrt{n})$.


\subsection{Outline}
The structure of the paper is as follows.
In Section \ref{sec:preliminaries} we describe some preliminaries and notations.
In Section \ref{sec:replacement-short} we apply our framework to the
replacement paths algorithm of Roditty and Zwick \cite{Roditty2005}.
In Section \ref{sec:dso-short} we apply our framework to the DSO of
Weimann and Yuster for graphs with real-edge weights \cite{WY13}.


In order for this paper to be self-contained, a full description of
the combinatorial deterministic replacement paths algorithm is given
in Section \ref{appendix:sec:replacement} and a full description of
the deterministic distance sensitivity oracles is given in Section
\ref{sec:dso1}.

\section{Preliminaries} \label{sec:preliminaries}
Let $G=(V,E)$ be a directed weighted graph with $n$ vertices and $m$
edges with real edge weights $\omega(\cdot)$. Given a path $P$ in
$G$ we define its weight $\omega(P) = \Sigma_{e \in E(P)}
\omega(e)$.

Given $s,t \in V$, let $P_G(s,t)$ be a shortest path from $s$ to
$t$ in $G$ and let $d_G(s,t) = \omega(P_G(s,t))$ be its length, which is the sum of its edge weights.
Let $|P_G(s,t)|$ denote the number of edges along $P_G(s,t)$.
Note that for unweighted graphs we have $|P_G(s,t)| = d_G(s,t)$.
When $G$ is known from the context we sometimes abbreviate $P_G(s,t), d_G(s,t)$ with $P(s,t), d(s,t)$ respectively.

We define the path concatenation operator $\circ$ as follows. Let
$P_1= (x_1,x_2, \ldots, x_r)$ and $P_2= (y_1,y_2, \ldots, y_t)$ be
two paths. Then $P = P_1 \circ P_2$ is defined as the path $P =
(x_1, x_2, \ldots, x_r, y_1, y_2, \ldots, y_t)$, and it is well
defined if either $x_r=y_1$ or $(x_r, y_1) \in E$.

For a graph $H$ we denote by $V(H)$ the set of its vertices, and by
$E(H)$ the set of its edges. When it is clear from the context, we
abbreviate $e \in E(H)$ by $e \in H$ and $v \in V(H)$ by $v \in H$.

Let $P$ be a path which contains the vertices $u,v \in V(P)$ such
that $u$ appears before $v$ along $P$. We denote by $P[u..v]$ the
subpath of $P$ from $u$ to $v$.

For every edge $e \in P_G(s,t)$ a replacement path $P_G(s,t,e)$ for the
triple $(s,t,e)$ is a shortest path from $s$ to $t$ avoiding $e$.
Let $d_G(s,t,e) = \omega(P_G(s,t,e))$ be the length of the replacement
path $P_G(s,t,e)$.

We will assume, without loss of generality, that every replacement
path $P_G(s,t,e)$ can be decomposed into a common prefix
$\CPref_{s,t,e}$ with the shortest path $P_G(s,t)$, a detour
$\Detour_{s,t,e}$ which is disjoint from the shortest path
$P_G(s,t)$ (except for its first vertex and last vertex), and finally a common suffix $\CSuff_{s,t,e}$ which is
common with the shortest path $P_G(s,t)$. Therefore, for every edge $e
\in P_G(s,t)$ it holds that $P_G(s,t,e) = \CPref_{s,t,e} \circ
\Detour_{s,t,e} \circ \CSuff_{s,t,e}$ (the prefix and/or suffix may
be empty).

Let $F \subseteq V \cup E$ be a set of vertices and edges. We define
the graph $G \setminus F = (V \setminus F, E \setminus F)$ as the
graph obtained from $G$ by removing the vertices and edges $F$. We
define a replacement path $P_G(s,t,F)$ as a
shortest path from $s$ to $t$ in the graph $G \setminus F$, and let
$d_G(s,t,F) = w(P_G(s,t,e))$ be its length.


\section{Deterministic Replacement Paths in $\Otilde(m \sqrt{n})$ Time - an Overview} \label{sec:replacement-short}
In this section we apply our framework from Section \ref{sec:framework} to the replacement paths algorithm of Roditty and Zwick \cite{Roditty2005}.
A full description of the deterministic replacement paths algorithm is given in Section \ref{appendix:sec:replacement}.

\shiri{I think this entire section is a bit detached... maybe write in few sentences what is the part in Roditty-Zwick that need to be derandomized and explain that this is what is being done in this section}
\sarel{I've added a few paragraphs here to address this remark}

The randomized algorithm by Roddity and Zwick as described in
\cite{Roditty2005} takes $\Otilde(m \sqrt{n})$ expected time.
They handle
separately the case that a replacement path has a short detour
containing at most $\lceil \sqrt{n} \rceil$ edges, and the case that a
replacement path has a long detour containing more than $\lceil \sqrt{n} \rceil$
edges. The first case is solved deterministically. The second
case is solved by first sampling a subset of vertices $R$ according to Lemma \ref{lem:sampling-roditty},
where each vertex is sampled uniformly independently at random with
probability $c \ln n/ \sqrt{n}$ for large enough constant $c > 0$.
Using this uniform sampling, it holds with high probability (of at least $1-n^{-c+2}$)
that for every long triple $(s,t,e)$ (as defined hereinafter), the detour $\Detour_{s,t,e}$ of the replacement path $P_G(s,t,e)$ contains at least one
vertex of $R$.

\begin{definition} \label{def:long-triple}
Let $s,t \in V, e \in P_G(s,t)$. The triple $(s,t,e)$ is a {\sl long} triple if every replacement path from $s$ to $t$ avoiding $e$ has its detour part containing more than $\lceil \sqrt{n} \rceil$ edges.
\end{definition}

Note that in Definition \ref{def:long-triple} we defined $(s,t,e)$ to be a long triple if {\bf every} replacement path from $s$ to $t$ avoiding $e$ has a long detour (containing more than $\lceil \sqrt{n} \rceil$ edges).
We could have defined $(s,t,e)$ to be a long triple even if at least one replacement path from $s$ to $t$ avoiding $e$ has a long detour (perhaps more similar to the definitions in \cite{Roditty2005}), however we find Definition \ref{def:long-triple} more convenient for the following reason.
If $(s,t,e)$ has a replacement path whose detour part contains at most $\lceil \sqrt{n} \rceil$ edges, then the algorithm of \cite{Roditty2005} for handling short detours finds deterministically a replacement path for $(s,t,e)$.
Hence, we only need to find the replacement paths for triples $(s,t,e)$ for which every replacement path from $s$ to $t$ avoiding $e$ has a long detour, and this is the case for which we define $(s,t,e)$ as a long triple.

It is sufficient for the correctness of the replacement paths
algorithm that the following condition holds; For every long triple
$(s,t,e)$ the detour $\Detour_{s,t,e}$ of the replacement path
$P_G(s,t,e)$ contains at least one vertex of $R$. As the authors of
\cite{Roditty2005} write, the choice of the random set $R$ is the
only randomization used in their algorithm. To obtain a
deterministic algorithm for the replacement paths problem and to
prove Theorem \ref{thm:replacement}, we prove the following
deterministic alternative of Lemma \ref{lemma:greedy}.

\begin{lemma} [Our derandomized version of Lemma \ref{lemma:greedy} for the replacement paths algorithm] \label{lemma:the-set-R-deterministic}
There exists an $\Otilde(m \sqrt{n})$ time deterministic algorithm that computes
a set $R \subseteq V$ of $\Otilde(\sqrt{n})$ vertices,
such that for every long triple $(s,t,e)$ there exists a replacement
path $P_G(s,t,e)$ whose detour part contains at least one of the
vertices of $R$.
\end{lemma}

Following the above description, in order to prove Theorem
\ref{thm:replacement}, that there exists an $\Otilde(m\sqrt{n})$
deterministic replacement paths algorithm, it is sufficient to prove
the derandomization Lemma \ref{lemma:the-set-R-deterministic}, we do
so in the following sections.

\subsection{Step 1: the Method of Reusing Common Subpaths - Defining the Set ${\cal D}_n$} \label{sec:reusing-common-subpaths}

In this section we prove the following lemma.

\begin{lemma} \label{thm:p-sqrt}
There exists a set ${\cal D}_{n}$ of at most $n$ paths, each path of
length exactly $\lceil \sqrt{n} \rceil$ with the following property;
for every long triple $(s,t,e)$ there exists a path $D \in {\cal
D}_{n}$ and a replacement path $P_G(s,t,e)$ such that $D$ is
contained in the detour part of $P_G(s,t,e)$.
\end{lemma}

In order to define the set of paths ${\cal D}_n$ and prove Lemma
\ref{thm:p-sqrt} we need the following definitions. Let  $G' = G
\setminus E(P_G(s,t))$ be the graph obtained by removing the edges
of the path $P_G(s,t)$ from $G$. For two vertices $u$ and $v$, let
$d_{G'}(u,v)$ be the distance from $u$ to $v$ in $G'$.

We use the following definitions of the index $\rho(x)$, the set of vertices $V_{\sqrt{n}}$ and the set of paths $\mathcal{D}_{n}$.

\begin{definition}[The index $\rho(x)$]\label{def:x-tag}
Let $P_G(s,t) = <v_0, \ldots, v_k>$ and let $X = \{ x \in V \ | \
\exists_{0 \le i \le k} \ d_{G'}(v_i, x) = \lceil \sqrt{n} \rceil
\}$ be the subset of all the vertices $x \in V$ such that there
exists at least one index $0 \le i \le k$ with $d_{G'}(v_i, x) =
\lceil \sqrt{n} \rceil$.

For every vertex $x \in X$ we define the index $0 \le \rho(x) \le k$ to be the minimum index such that $d_{G'}(v_{\rho(x)}, x) = \lceil \sqrt{n} \rceil$.
\end{definition}

\begin{definition} [The set of vertices $V_{\sqrt{n}}$] \label{def:v-sqrtn}
We define the set of vertices $V_{\sqrt{n}} = \{ x \in X | \forall_{i < \rho(x)} d_{G'}(v_i, x) > \lceil \sqrt{n} \rceil \}$.
In other words, $V_{\sqrt{n}}$ is the set of all vertices $x \in X$
such that for all
the vertices $v_i$ before $v_{\rho(x)}$ along $P_G(s,t)$ it holds that
$d_{G'}(v_i, x) > \lceil \sqrt{n} \rceil$.
\end{definition}

\begin{definition} [A set of paths $\mathcal{D}_{n}$] \label{def:p-sqrtn}
For every vertex $x \in
V_{\sqrt{n}}$, let $D(x)$ be an arbitrary shortest path from $v_{\rho(x)}$ to $x$ in
$G'$ (whose length is $\lceil \sqrt{n} \rceil$ as $d_{G'}(v_{\rho(x)}, x) = \lceil \sqrt{n} \rceil$).
We define ${\cal D}_{n} = \{ D(x) | x \in V_{\sqrt{n}} \}$.
\end{definition}

Note that while $V_{\sqrt{n}}$ is uniquely defined (as it is defined
according to distances between vertices) the set of paths
$\mathcal{D}_{n}$ is not unique, as there may be many shortest paths
from $v_{\rho(x)}$ to $x$ in $G'$, and we take $D(x) =
P_{G'}(v_{\rho(x)}, x)$ to be an arbitrary such shortest path.

The basic intuition for the method of reusing common subpaths is
as follows. Let $P_G(s,t,e_1), \ldots, P_G(s,t,e_r)$ be arbitrary replacement paths such that $x$
is the $(\lceil \sqrt{n} \rceil + 1)^{\text{th}}$
vertex along the detours of all the replacement path  $P_G(s,t,e_1), \ldots, P_G(s,t,e_r)$.
Then one can construct replacement paths $P'_G(s,t,e_1), \ldots, P'_G(s,t,e_r)$ such that the subpath $D(x) \in {\cal D}_n$ is contained in all these replacement paths.
Therefore, the subpath $D(x)$ is reused as a common subpath in many replacement paths.
We utilize this observation in the following proof
of Lemma \ref{thm:p-sqrt}.


\begin{proof} [Proof of Lemma \ref{thm:p-sqrt}]
Obviously, the set ${\cal D}_{n}$ described in Definition \ref{def:p-sqrtn} contains at most
$n$ paths, each path is of length exactly $\lceil \sqrt{n} \rceil$.

We prove that for every long triple $(s,t,e)$ there exists a path $D \in \mathcal{D}_{n}$
and a replacement path $P'(s,t,e)$ {\sl s.t.} $D$ is contained in the detour part of $P'(s,t,e)$.

Let $P_G(s,t,e)$ be a replacement path for $(s,t,e)$. Since $(s,t,e)$ is
a long triple then the detour part $\Detour_{s,t,e}$ of $P_G(s,t,e)$
contains more than $\lceil \sqrt{n} \rceil$ edges. Let $x \in \Detour_{s,t,e}$
be the $(\lceil \sqrt{n} \rceil + 1)^{\text{th}}$ vertex along $\Detour_{s,t,e}$, and
let $v_j$ be the first vertex of $\Detour_{s,t,e}$. Let $P_1$ be
the subpath of $\Detour_{s,t,e}$ from $v_j$ to $x$ and let $P_2$
be the subpath of $P_G(s,t,e)$ from $x$ to $t$. In other words,
$P_G(s,t,e) = <v_0, \ldots, v_j> \circ P_1 \circ P_2$. Since
$\Detour_{s,t,e}$ contains more than $\lceil \sqrt{n} \rceil$ edges and is disjoint
from $P_G(s,t)$ except for the first and last vertices of $\Detour_{s,t,e}$
and $P_1 \subset \Detour_{s,t,e}$ it follows that
$P_1$ is disjoint from $P_G(s,t)$ (except for the vertex $v_j$). In
particular,
since $P_1$ is a shortest path in $G\setminus \{e\}$ that is edge-disjoint from $P_G(s,t)$, then $P_1$ is also a shortest path in $G' = G \setminus E(P_G(s,t))$.
We get that $d_{G'}(v_j,x) = |P_1| = \lceil \sqrt{n} \rceil$.

We prove that $j=\rho(x)$ and $x \in V_{\sqrt{n}}$. As we have
already proved that $d_{G'}(v_j, x) = \lceil \sqrt{n} \rceil$, we
need to prove that for every $0 \le i < j$ it holds that
$d_{G'}(v_i, x) > \lceil \sqrt{n} \rceil$. Assume by contradiction
that there exists an index $0 \le i < j$ such that $d_{G'}(v_i, x)
\le \lceil \sqrt{n} \rceil$. Then the path $\hat{P} = <v_0, \ldots,
v_i> \circ P_{G'}(v_i, x) \circ P_2$ is a path from $s$ to $t$ that
avoids $e$ and its length is:

\begin{eqnarray*}
  |\hat{P}| &=& |<v_0, \ldots, v_i> \circ P_{G'}(v_i, x) \circ P_2| \\
   & \le & i + \lceil \sqrt{n} \rceil + |P_2| \\
   & < & j + \lceil \sqrt{n} \rceil + |P_2| \\
   & = & |P_G(s,v_j) \circ P_1 \circ P_2| \\
   & = & |P_G(s,t,e)|
\end{eqnarray*}

This means that the path $\hat{P}$ is a path from $s$ to $t$ in
$G\setminus \{e\}$ and its length is shorter than the length of the
shortest  path  $P_G(s,t,e)$ from $s$ to $t$ in $G\setminus\{e\}$, which
is a contradiction. We get that $d_{G'}(v_j,x) = \lceil \sqrt{n} \rceil$ and
for every $0 \le i < j$ it holds that $d_{G'}(v_i, x) > \lceil \sqrt{n} \rceil$.
Therefore, according to Definitions
\ref{def:x-tag} and \ref{def:v-sqrtn} it holds that
$j=\rho(x)$ and $x \in V_{\sqrt{n}}$.

Let $D(x) \in \mathcal{D}_{n}$, then according to Definition
\ref{def:p-sqrtn}, $D(x)$ is a shortest path from $v_{\rho(x)}$ to
$x$ in $G'$. We define the path $P'(s,t,e) = <v_0, \ldots,
v_{\rho(x)}> \circ D(x) \circ P_2$. It follows that $P'(s,t,e)$ is a
path from $s$ to $t$ that avoids $e$ and $|P'(s,t,e)| = |<v_0,
\ldots, v_{\rho(x)}> \circ D(x) \circ P_2| = \rho(x) + \lceil
\sqrt{n} \rceil + |P_2| = |P_G(s,t,e)| = d_G(s,t,e)$. Hence, $P'(s,t,e)$
is a replacement path for $(s,t,e)$ such that $D(x) \subset
P'(s,t,e)$ so the lemma follows.
\end{proof}

\subsection{Step 2: the Method of Decremental Distances from a Path - Computing the Set ${\cal D}_n$} \label{sec:rp-efficient}

In this section we describe a decremental algorithm
that enables us to compute the set of paths ${\cal D}_{n}$ in
$\Otilde(m\sqrt{n})$ time, proving the following lemma.

\begin{lemma} \label{lemma:compute-dn}
There exists a deterministic algorithm for computing the set of
paths ${\cal D}_{n}$ in $\Otilde(m\sqrt{n})$ time.
\end{lemma}

Our algorithm for computing the set of path ${\cal D}_{n}$ is a
variant of the decremental SSSP (single source shortest paths)
algorithm of King \cite{King99}. Our variant of the algorithm is
used to find distances of vertices from a path rather than from a
single source vertex as we define below.


{\bf Overview of the Deterministic Algorithm for Computing ${\cal D}_{n}$ in $\Otilde(m\sqrt{n})$ Time.}

In the following description let $P = P_G(s,t)$.
Consider the following assignment of weights $\omega$ to edges of $G$. We
assign weight $\epsilon$ for every edge $e$ on the path $P$, and
weight $1$ for all the other edges where $\epsilon$ is a small number such that $0 < \epsilon < 1/n$. We
define a graph $G^w = (G,w)$ as the weighted graph $G$ with edge
weights $\omega$.
We define for every $0 \le i \le k$ the graph $G_i = G
\setminus \{v_{i+1}, \ldots, v_k\}$ and the path $P_i = P \setminus
\{v_{i+1}, \ldots, v_k\}$. We define the graph $G^w_i = (G_i, w)$ as
the weighted graph $G_i$ with edge weights $\omega$.

The algorithm computes the graph $G^w$ by simply taking $G$ and
setting all edge weights of $P_G(s,t)$ to be $\epsilon$ (for some
small $\epsilon$ such that $\epsilon < 1/n$) and all other edge
weights to be 1. The algorithm then removes the vertices of
$P_G(s,t)$ from $G^w$ one after the other (starting from the vertex
that is closest to $t$). Loosely speaking after each vertex is
removed, the algorithm computes the distances from $s$ in the
current graph. In each such iteration, the algorithm adds to
$V^w_{\sqrt{n}}$ all vertices such that their distance from $s$ in
the current graph is between $\lceil \sqrt{n} \rceil$ and $\lceil
\sqrt{n} \rceil+1$. We will later show that at the end of the
algorithm we have $V^w_{\sqrt{n}} = V_{\sqrt{n}}$. Unfortunately, we
cannot afford running Dijkstra after the removal of every vertex of
$P_G(s,t)$ as there might be $n$ vertices on $P_G(s,t)$. To overcome
this issue, the algorithm only maintains nodes at distance at most
$\lceil \sqrt{n} \rceil +1$ from $s$. In addition, we observe that
to compute the SSSP from $s$ in the graph after the removal of a
vertex $v_i$ we only need to spend time on nodes such that their
shortest path from $s$ uses the removed vertex. Roughly speaking,
for these nodes we show that their distance from $s$ rounded down to
the closest integer must increase by at least 1 as a result of the
removal of the vertex. Hence, for every node we spend time on it in
at most $\lceil \sqrt{n} \rceil+1$ iterations until its distance
from $s$ is bigger than $\lceil \sqrt{n} \rceil+1$. As we will show
later this will yield our desired running time.

In Section \ref{sec:p-sqrt-computation} we give a formal description and analysis of the algorithm and prove Lemma \ref{lemma:compute-dn}.

{\bf Proof of Theorem \ref{thm:replacement}.}
We summarize the $\Otilde(m \sqrt{n})$ deterministic replacement
paths algorithm and outline the proof of Theorem
\ref{thm:replacement}.
First, compute in $\Otilde(m\sqrt{n})$ time the set of paths ${\cal
D}_{n}$ as in Lemma \ref{lemma:compute-dn}.
Given ${\cal D}_{n}$, the deterministic greedy selection algorithm
GreedyPivotsSelection$(\mathcal{D}_{n})$ (as described in
Lemma \ref{lemma:greedy}) computes a set $R \subset
V$ of $\Otilde(\sqrt{n})$ vertices in $\Otilde( n\sqrt{n})$ time
with the following property; every path $D \in \mathcal{D}_{n}$
contains at least one of the vertices of $R$. Theorem
\ref{thm:replacement} follows from Lemmas
\ref{lemma:the-set-R-deterministic}, \ref{thm:p-sqrt} and
\ref{lemma:compute-dn}.


\section{Deterministic Distance Sensitivity Oracles - an Overview} \label{sec:dso-short}
In this section we apply our framework from Section \ref{sec:framework} to the combinatorial distance sensitivity oracles of Weimann and Yuster \cite{WY13}.
A full description of the deterministic combinatorial distance sensitivity oracles is given in Section \ref{sec:dso1}.

Let $0 < \epsilon < 1$ and $1 \le f = O(\frac{\log n}{\log \log n})$
be two parameters.
In \cite{WY13}, Weimann and Yuster considered the following notion
of intervals (note that in \cite{WY13} they use a parameter $0 <
\alpha < 1$ and we use a parameter $0 < \epsilon < 1$ such that
$\epsilon = 1 - \alpha$). They define an interval of a long simple
path $P$ as a subpath of $P$ consisting of $n^{\epsilon/f}$
consecutive vertices, so every simple path induces less than $n$
(overlapping) intervals. For every subset $F \subset E$ of at most
$f$ edges, and for every pair of vertices $u, v \in V$, let
$P_G(u,v,F)$ be a shortest path from $u$ to $v$ in $G \setminus F$.
The path $P_G(u,v,F)$ induces less than $n$ (overlapping) intervals.
The total number of possible intervals is less than $O(n^{2f+3})$ as each one of the (at most) $O(n^{2f+2})$
possible queries $(u,v,F)$ corresponds to a shortest path $P_G(u,v,F)$ that induces less than $n$ intervals.

\begin{definition}
Let ${\cal D}_{f}$ be defined as all the intervals (subpaths containing $n^{\epsilon/f}$ edges)
of all the replacement paths $P_G(s,t,F)$ for every $s,t \in V, F \subseteq E \cup V$ with $|F| \le f$.
\end{definition}

Weimann and Yuster apply Lemma \ref{lem:sampling-roditty} to find a set $R \subseteq V$ of $\Otilde(n^{1- \epsilon/f})$ vertices that hit w.h.p. all the intervals ${\cal D}_{f}$. 
According to these bounds (that ${\cal D}_{f}$ contains $O(n^{2f+3})$ paths, each containing exactly $n^{\epsilon/f}$ edges) applying the greedy algorithm to obtain the set $R$ deterministically according to Lemma \ref{lemma:greedy} takes $\Otilde(n^{2f+3 + \epsilon/f})$ time, which is very inefficient.

In this section
we assume that all weights are non-negative (so we can run Dijkstra's algorithm) and that shortest paths are unique, we justify these assumptions in Section \ref{sec:assumption-unique}.


\subsection{Step 1: the Method of Using Fault-Tolerant Trees to Significantly Reduce the Number of Intervals} \label{sec:ft-trees}
In Lemma \ref{lemma:small-num-of-intervals} we prove that the set of
intervals ${\cal D}_{f}$ actually contains at most
$O(n^{2+\epsilon})$ unique intervals, rather than the $O(n^{2f+3})$
naive upper bound mentioned above.
From Lemmas \ref{lemma:small-num-of-intervals} and
\ref{lemma:greedy} it follows that the
GreedyPivotsSelection$(\mathcal{D}_f)$ finds in
$\Otilde(n^{2+\epsilon+\epsilon/f})$ time the subset $R \subseteq V$
of $\Otilde(n^{1-\epsilon/f})$ vertices that hit all the intervals
${\cal D}_f$. In Section \ref{sec:improved-greedy} we further reduce
the time it takes for the greedy algorithm to compute the set of
pivots $R$ to $\Otilde(n^{2+\epsilon})$.

\begin{lemma} \label{lemma:small-num-of-intervals}
$|{\cal D}_{f}| = O(n^{2+\epsilon})$.
\end{lemma}

%

In order to prove Lemma \ref{lemma:small-num-of-intervals} we describe the fault-tolerant trees data-structure,
which is a variant of the trees which appear in Appendix A of
\cite{ChCoFiKa17}.

\begin{definition}
Let $P^L_G(s,t, F)$ be the shortest among the $s$-to-$t$ paths in $G \setminus F$ that contain at most $L$ edges and let $d^L_G(s,t,F) = \omega(P^L_G(s,t,F))$.
In other words, $d^L_G(s,t,F) = \min \{ \omega(P) \ | \ P \text{ is an } s-\text{to}-t \text{ path on at most } L \text{ edges} \}$.
If there is no path from $s$ to $t$ in $G\setminus F$ containing at most $L$ edges then we define $P^L_G(s,t, F) = \emptyset$ and $d^L_G(s,t,F) = \infty$.
For $F = \emptyset$ we abbreviate $P^L_G(s,t, \emptyset) = P^L_G(s,t)$ as the shortest path from $s$ to $t$ that contains at most $L$ edges, and $d^L_G(s,t) = d^L_G(s,t, \emptyset)$ as its length.
\end{definition}

Let $s,t \in V$ be vertices and let $L, f\ge 1$ be fixed integer
parameters, we define the trees $FT^{L,f}(s,t)$ as follows.
\begin{itemize}
\item In the root of $FT^{L,f}(s,t)$ we store the path $P^L_G(s,t)$ (and its length $d^L_G(s,t)$), and also store the vertices and edges of $P^L_G(s,t)$ in a binary search tree $BST^L(s,t)$; If $P^L_G(s,t) = \emptyset$ then we
terminate the construction of $FT^{L,f}(s,t)$.
\item For every edge or vertex $a_1$ of $P^L_G(s,t)$ we recursively build a subtree $FT^{L,f}(s,t, a_1)$ as follows.
Let $P^L_G(s,t, \{a_1\})$ be the shortest path from $s$ to $t$ that
contains at most $L$ edges in the graph $G\setminus \{a_1\}$. Then
in the subtree $FT^{L,f}(s,t,a_1)$ we store the path $P^L_G(s,t,
\{a_1\})$ (and its length $d^L_G(s,t, \{a_1\})$) and we also store the
vertices and edges of $P^L_G(s,t, \{a_1\})$ in a binary search tree
$BST^L(s,t,a_1)$; If $P^L_G(s,t, \{a_1\}) = \emptyset$ we terminate
the construction of $FT^{L,f}(s,t,a_1)$. If $f > 1$ then for every
vertex or edge $a_{2}$ in $P^L_G(s,t, \{a_1\})$ we recursively build
the subtree $FT^{L,f}(s,t, a_1, a_2)$ as follows.
\item For the recursive step, assume we want to construct the subtree $FT^{L,f}(s,t, a_1, \ldots, a_i)$.
In the root of $FT^{L,f}(s,t, a_1, \ldots, a_i)$ we store the path
$P^L_G(s,t, \{a_1, \ldots, a_i\})$ (and its length $d^L_G(s,t, \{a_1,
\ldots, a_i\})$) and we also store the vertices and edges of
$P^L_G(s,t, \{a_1, \ldots, a_i\})$ in a binary search tree $BST^L(s,t,
a_1, \ldots, a_i)$. If $P^L_G(s,t, \{a_1, \ldots, a_i\}) = \emptyset$
then we terminate the construction of $FT^{L,f}(s,t, a_1, \ldots,
a_i)$. If $i< f$ then for every vertex or edge $a_{i+1}$ in
$P^L_G(s,t, \{a_1, \ldots, a_i\}))$ we recursively build the subtree
$FT^{L,f}(s,t, a_1, \ldots, a_i, a_{i+1})$.
\end{itemize}

Observe that there are two conditions in which we terminate the
recursive construction of $FT^{L,f}(s,t, a_1, \ldots, a_i)$:
\begin{itemize}
\item Either $i = f$ in which case $FT^{L,f}(s,t, a_1, \ldots, a_f)$ is a leaf node of $FT^{L,f}(s,t)$ and we store in the leaf node $FT^{L,f}(s,t, a_1, \ldots, a_f)$ the path $P^L_G(s,t, \{a_1, \ldots, a_f\})$.
\item Or there is no path from $s$ to $t$ in $G \setminus \{a_1, \ldots, a_i \}$ that contains at most $L$ edges and then  $FT^{L,f}(s,t, a_1, \ldots, a_i)$ is a leaf vertex of $FT^{L,f}(s,v)$ and we store in it $P^L_G(s,t, \{a_1, \ldots, a_i\}) = \emptyset$.
\end{itemize}


{\bf Querying the tree $FT^{L,f}(s,t)$.} Given a query $(s,t,F)$
such that $F \subset V \cup E$ with $|F| = f$  we would like to
compute $d^L_G(s,t,F)$ using the tree $FT^{L,f}(s,t)$.

The query procedure is as follows. Let $P^L_G(s,t)$ be the path stored
in the root of $FT^{L,f}(s,t)$ (if the root of $FT^{L,f}(s,t)$
contains $\emptyset$ then we output that $d^L_G(s,t,F) = \infty$).
First we check if  $P^L_G(s,t) \cap F = \emptyset$ by checking if any
of the elements $a_1 \in F$ appear in $BST^L(s,t)$ (which takes
$O(\log L)$ time for each element $a_1 \in F$). If $P^L_G(s,t) \cap F
= \emptyset$ we output $d^L_G(s,t,F) = d^L_G(s,t)$ (as $P^L_G(s,t)$ does
not contain any of the vertices or edges in $F$). Otherwise, let
$a_1 \in P^L_G(s,t) \cap F$.

We continue the search similarly in the subtree $FT^{L,f}(s,t, a_1)$
as follows. Let $P^L_G(s,t, \{a_1\})$ be the path stored in the root
of $FT^{L,f}(s,t, a_1)$ (if the root of $FT^{L,f}(s,t, a_1)$
contains $\emptyset$ then we output that $d^L_G(s,t,F) = \infty$).
First we check if  $P^L_G(s,t, \{a_1\}) \cap F = \emptyset$ by
checking if any of the elements $a_2 \in F$ appear in   $BST^L(s,t,
a_1)$ (which takes $O(\log L)$ time for each element $a_2 \in F$).
If $P^L_G(s,t, \{a_1\}) \cap F = \emptyset$ we output $d^L_G(s,t,F) =
d^L_G(s,t, \{a_1\})$ (as $P^L_G(s,t, \{a_1\})$ does not contain any of
the vertices or edges in $F$). Otherwise, let $a_2 \in P^L_G(s,t,
\{a_1\}) \cap F$.
We continue the search similarly in the subtrees $FT^{L,f}(s,t, a_1,
a_2)$, $FT^{L,f}(s,t, a_1, a_2, \ldots, a_i)$ until we either reach
a leaf node which contains $\emptyset$ (and in this case we output
that $d^L_G(s,t,F) = \infty$) or we find a path $P^L_G(s,t, \{a_1,
\ldots, a_i\})$ such that $P^L_G(s,t, \{a_1, \ldots, a_i\}) \cap F =
\emptyset$ and then we output $d^L_G(s,t,F) = d^L_G(s,t, \{a_1, \ldots,
a_i\})$.

In Section \ref{sec:ft-trees-appendix} we prove the following lemma.

\begin{lemma} \label{lemma:query}
Given the tree $FT^{L,f}(s,t)$ and a set of failures $F \subset V
\cup E$ with $|F| \le f$, the query procedure computes the distance
$d^L_G(s,t,F)$ in $O(f^2 \log L)$ time.
\end{lemma}

We are now ready to prove lemma \ref{lemma:small-num-of-intervals} asserting that $|{\cal D}_{f}| = O(n^{2+\epsilon})$.

\begin{proof} [Proof of Lemma \ref{lemma:small-num-of-intervals}]
Let $L=n^{\epsilon/f}$ and let ${\cal D}$ be the set of all the
unique shortest paths $P^L_G(s,t, \{a_1, \ldots, a_i\})$ stored in
all the nodes of all the trees $\{ FT^{L,f}(s,t) \}_{s,t \in V}$
(see Section \ref{sec:assumption-unique} for more details on the
assumption of unique shortest paths in our algorithms). Since the
number of nodes in every tree $FT^{L,f}(s,t)$ is at most $L^f =
(n^{\epsilon/f})^f = n^{\epsilon}$, and there are $O(n^2)$ trees
(one tree for every pair of vertices $s,t \in V$) we get that the
number of nodes in all the trees $\{ FT^{L,f}(s,t) \}_{s,t \in V}$
is $O(n^{2+\epsilon})$ and hence $|{\cal D}| = O(n^{2+\epsilon})$.

We prove that ${\cal D}_{f} \subseteq {\cal D}$. By definition,
${\cal D}_{f}$ contains all the intervals (subpaths containing
$n^{\epsilon/f}$ edges) of all the replacement paths $P_G(s,t,F)$ for
every $s,t \in V, F \subseteq E \cup V$ with $|F| \le f$. Let $P \in
{\cal D}_{f}$ be the unique shortest path as defined in Section
\ref{sec:assumption-unique}, then $P$ is a subpath containing
$n^{\epsilon/f}$ edges of the replacement paths $P_G(s,t,F)$. Let $u$
be the first vertex of $P$, and let $v$ be the last vertex of $P$.
Then $P$ is a shortest path from $u$ to $v$ in $G\setminus F$, and
since we assume that the shortest paths our algorithms compute are
unique (according to Section \ref{sec:assumption-unique}) then $P =
P_G(u,v,F)$ is the unique shortest path from $u$ to $v$ in $G\setminus
F$. Since $P$ is assumed to be a path on exactly $L =
n^{\epsilon/f}$ edges, then $P = P_G(u,v,F) = P^L_G(u,v,F)$. According
to the query procedure in the tree $FT^{L,f}(u,v)$ and Lemma
\ref{lemma:query}, if we query the tree $FT^{L,f}(u,v)$ with
$(u,v,F)$ then we reach a node $FT^{L,f}(u,v, a_1, \ldots, a_i)$
which contains the path $P^L_G(u,v, \{a_1, \ldots, a_i\})$ with
$\{a_1, \ldots, a_i \} \subseteq F$ such that $P^L_G(u,v, \{a_1,
\ldots, a_i\}) = P^L_G(u,v,F) = P$ is the shortest $u$-to-$v$ path in
$G\setminus F$. Hence, $P \in {\cal D}$ and thus ${\cal D}_{f}
\subseteq {\cal D}$ and $|{\cal D}_{f}| \le |{\cal D}| =
O(n^{2+\epsilon})$
\end{proof}

\subsection{Step 2: Efficient Construction of the Fault-Tolerant Trees - Computing the Paths ${\cal D}_f$} \label{sec:dynamic-programming}
Recall that we defined the trees $FT^{L,f}(u,v)$ with respect the
parameters $f$ (the maximum number of failures)  and $L$ (where we
search for shortest paths among paths of at most $L$ edges).
The idea is to build the trees $FT^{L,f}(u,v)$ using dynamic programming having the trees $FT^{L-1, f}(u,v)$ with parameters $f,
L-1$ as subproblems.

Assume we have already built the trees $FT^{i,f}(u,v)$, where $u,v
\in V, 1 \le i < L$, we describe how to build the trees
$FT^{i+1,f}(u,v)$.
Let $(u,v,F)$ be a query for which we want to compute the distance
$d^{i+1}(u,v,F)$ (as part of the construction of the tree
$FT^{i+1,f}(u,v)$). Scan all the edges $(u,z) \in E$ and query the
tree $FT^{i, f}(z,v)$ with the set $F$ to find the distance
$d^i(z,v,F)$. Querying the tree $FT^{i,f}(z,v)$ takes $O(f^2 \log i)
= O(f^2 \log L)$ time as described in Lemma \ref{lemma:query} (note
that $f^2 \log L = \Otilde(1)$ for $f \le \log n$ as  $L \le n$),
and we run $O(\text{out-degree(u)})$ such queries and take the
minimum of the following equation.
\begin{eqnarray} \label{eq:dynamic-programming}
  d^{i+1}(u,v, F)  = \min_{z} \{ \omega(u,z) + d^i(z,v,F) \ | \  (u,z) \in E \ \ AND \ \ u,z,(u,z) \not \in F
  \}
\end{eqnarray} \begin{eqnarray}  \label{eq:parent-pointer}
  \textrm{parent}^{i+1}(u,v, F) = \arg \min_{z} \{ \omega(u,z) + d^i(z,v,F) \ | \  (u,z) \in E \ \ AND \ \ u,z,(u,z) \not \in F
  \}
\end{eqnarray}


Note that in Equation \ref{eq:dynamic-programming} we assume that for every vertex $u \in V$ it holds that $G$ contains the self loops $(u,u) \in E$
such that $\omega(u,u) = 0$.

So the time to compute $d^{i+1}(u,v,F)$ is
$\Otilde(\text{out-degree(u)})$. Next, we describe how to
reconstruct the path $P^{i+1}(u,v,F)$ in $O(L)$ additional time. We
reconstruct the shortest path $P^{i+1}(u,v,F)$ by simply following
the (at most $L$) parent pointers. In more details, let $z =
\textrm{parent}^{i+1}(u,v, F)$ be the vertex defined according to
Equation \ref{eq:parent-pointer}. We reconstruct the shortest path
$P^{i+1}(u,v,F)$ by concatenating $(u,z)$ with the shortest path
$P^{i}(z,v,F)$ (which we reconstruct in the same way), thus we can
reconstruct $P^{i+1}(u,v,F)$ edge by edge in constant time per edge,
and hence it takes $O(L)$ time to reconstruct the path
$P^{i+1}(u,v,F)$ that contains at most $L$ edges.

The tree $FT^{i,f}(u,v)$ contains $i^f \le L^f$ nodes, and thus all the trees $\{FT^{i,f}(u,v)\}$ for all $i \le L, u,v \in V$ contain
$O(n^2 L^{f+1})$ nodes together.

In each such node we compute the distance $d^i(u,v, \{a_1, \ldots, a_j\})$ in $\Otilde(\text{out-degree(u)})$ time and reconstruct the path $P^i(u,v, \{a_1, \ldots, a_j\})$ in additional $O(L)$ time.
Theretofore, computing all the distances $d^i(u,v, \{a_1, \ldots, a_j\})$ and all the paths $P^i(u,v, \{a_1, \ldots, a_j\})$
in all the nodes of all the trees $\{FT^{i,f}(u,v) \}_{u,v \in V, 1 \le i \le L}$
takes $\Otilde(\sum_{i \le L, u,v \in V} {L^f(\text{out-degree(u)}+L)}) = \Otilde(mnL^{f+1} + n^2L^{f+2})$ time.
substituting $L = \Otilde(n^{\epsilon/f})$ we get an algorithm to compute the trees $\{FT^{L,f}(u,v) \}_{u,v \in V}$ in $\Otilde(mn^{1+\epsilon+\epsilon/f} + n^{2+\epsilon+2\epsilon/f})$ time.

This proves the following Lemma.

\begin{lemma} \label{lemma:ft-tree-eps-plus}
One can deterministically construct the trees $FT^{L,f}(s,t)$ for every $s,t \in V$ in $\Otilde(mn^{1+\epsilon+\epsilon/f} + n^{2+\epsilon+2\epsilon/f})$ time.
\end{lemma}

In Section \ref{sec:positive} we further reduce the runtime to
$\Otilde(mn^{1+\epsilon})$ by using dynamic programming only for
computing the first $f-1$ levels of the trees $FT^{L,f}(s,t)$ and
then applying Dijkstra in a sophisticated manner to compute the last
layer of the trees $FT^{L,f}(s,t)$. In addition, we also boost-up
the runtime of the greedy pivots selection algorithm from
$\Otilde(n^{2+\epsilon +\epsilon/f})$
to $\Otilde(n^{2+\epsilon})$ time.

\section{Deterministic Replacement Paths Algorithm} \label{appendix:sec:replacement}
In this section we add the missing parts of the $\Otilde(m\sqrt{n})$ time deterministic replacement paths algorithm,
derandomizing the replacement paths algorithm of Roddity and Zwick \cite{Roditty2005}.
Recall the notion of a long triple $(s,t,e)$ as in Definition \ref{def:long-triple}.
Let $s,t \in V, e \in P_G(s,t)$, the triple $(s,t,e)$ is a {\sl long} triple if for every replacement path from $s$ to $t$ avoiding $e$ has its detour part containing more than $\lceil \sqrt{n} \rceil$ edges.

In order for this paper to be self-contained, let us start by describing the randomized $\Otilde(m\sqrt{n})$ replacement paths algorithm of Roditty and Zwick \cite{Roditty2005}.

\subsection{The Randomized $\Otilde(m\sqrt{n})$ Replacement Paths Algorithm of Roditty and Zwick - a Summary} \label{appendix:sec:rodittys-algorithm}
The algorithm by Roddity and Zwick that is described in
\cite{Roditty2005} takes $\Otilde(m \sqrt{n})$ time. Their algorithm handles
separately the case that a replacement path has a short detour
containing at most $\lceil \sqrt{n} \rceil$ edges, and the case that a
replacement path has a long detour containing more than $\lceil \sqrt{n} \rceil$
edges. The first case is solved deterministically while the second
case is solved by a randomized algorithm as described below.


\subsubsection{Handling Short Detours}
Roditty and Zwick's algorithm finds replacement paths with short detours (containing at most $\lceil \sqrt{n} \rceil$ edges)
deterministically. Let $P = P_G(s,t) = <v_0, \ldots, v_k>$ be the shortest path from $s$ to $t$,
and let $G' = G \setminus E(P_G(s,t))$ be the graph $G$ after removing the edges of $P$ and let $\ell = \lfloor \frac{k}{2\sqrt{n}} \rfloor$.

As explained in Section \ref{sec:preliminaries}, for every triple $(s,t,e) \in V \times V \times E$ every replacement paths $P_G(s,t,e)$
can be partitioned into a common prefix $\CPref_{s,t,e}$, a disjoint detour $\Detour_{s,t,e}$ and a common suffix $\CSuff_{s,t,e}$.
In this part of handling short detours, we would like to find all distances $d_G(s,t,e)$ such that there exists at least one replacement path $P_G(s,t,e)$ whose detour part $\Detour_{s,t,e}$ contains at most $\lceil \sqrt{n} \rceil$ edges.

The algorithm for handling short detours has two parts. The first part, computes a table $RD[i,j]$ which is defined as follows.
For every $0 \le i \le k$ and $0 \le j \le \lceil \sqrt{n} \rceil -1$ the entry $RD[i,j]$ gives the length of the
shortest path in $G'$ ({\sl i.e.}, the detour) starting at $v_{i}$ and ending at $v_{i+j}$, if its length is at most $\lceil \sqrt{n} \rceil$, or indicates that $d_{G'}(v_{i}, v_{i+j}) > \lceil \sqrt{n} \rceil$.
The second part, uses the table of detours $RD$ to find replacement paths whose detour part contains at most $\lceil \sqrt{n} \rceil$ edges.

{\bf First part: computing the table $RD$.}
The algorithm builds an auxiliary graph $G^A$ obtained by adding a new source vertex $r$ to $G'$
and an edge $(r, v_{2q \lceil \sqrt{n} \rceil})$ of weight $\omega(r, v_{2q \lceil \sqrt{n} \rceil})=q \lceil \sqrt{n} \rceil$ for every $0 \le q \le \ell$.
The weight of all the edges $E \setminus E(P)$ is set to $1$. Then the algorithm runs Dijkstra's algorithm
from $r$ in $G^A$ to find all the best short detours ({\sl i.e.}, the shortest paths on at most $\lceil \sqrt{n} \rceil$ edges from $v_i$ to $v_j$ in $G'$ where $1 \le i < j \le k$) that start in one of the vertices
$v_0, v_{2 \lceil \sqrt{n} \rceil}, \ldots, v_{2\ell \lceil \sqrt{n} \rceil}$.
See Figure \ref{fig:the-graph-G-A} as an illustration of the graph $G^A$.

\begin{figure}[htb]
  \centering
  \includegraphics[width=1.0\linewidth]{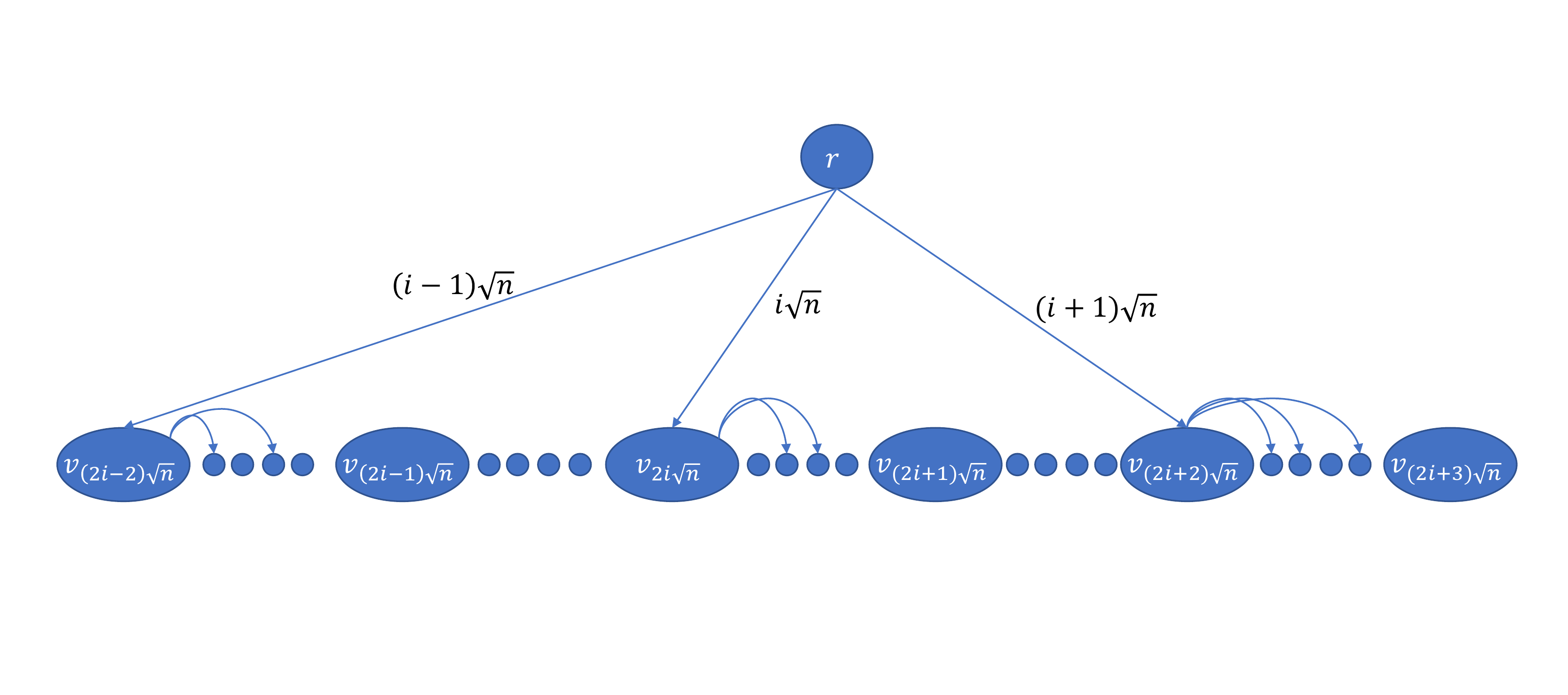}
  \caption{An illustration of the auxiliary graph $G^A$ obtained from $G' = G \setminus E(P)$ by adding the following auxiliary edges. We add the edges $(r, v_{2q \lceil \sqrt{n} \rceil})$ of weight $q\lceil \sqrt{n} \rceil$ for every $0 \le q \le \ell$. In this example we
    present the subgraph with the vertices $v_{(2q-2) \lceil \sqrt{n} \rceil}, \ldots, v_{(2q+3)\lceil \sqrt{n} \rceil}$, where we add the edges $(r,v_{(2q-2)\lceil \sqrt{n} \rceil}), (r,v_{2q \lceil \sqrt{n} \rceil}), (r,v_{(2q+2)\lceil\sqrt{n} \rceil})$ with weights $(q-1)\lceil \sqrt{n} \rceil, q\lceil\sqrt{n}\rceil, (q+1) \lceil\sqrt{n} \rceil$ respectively. }
  \label{fig:the-graph-G-A}
\end{figure}

In a sense, the algorithm already found all the relevant detours from about $\frac{k}{\sqrt{n}}$ of the vertices.
More precisely, the algorithm has found the entries $RD[i,j]$ for $0 \le i \le k, 0 \le j \le \lceil \sqrt{n} \rceil -1$ such that $(i \mod 2 \lceil \sqrt{n} \rceil) = 0$. By running this algorithm in $O(\sqrt{n})$ phases (in phase $p$ we compute the entries of $RD[i,j]$ such that $0 \le i \le k, 0 \le j \le \lceil \sqrt{n} \rceil-1$ and $(i \mod 2\lceil \sqrt{n} \rceil) = p$), we can find all the relevant detours starting from all the nodes of the path $P$.
That is, we run this algorithm
$2\lceil \sqrt{n} \rceil-1$ more times to find short detours emanating from the other vertices of $P_G(s,t)$.
In the $p^\text{th}$ phase (for $0 \le p \le 2\lceil \sqrt{n} \rceil-1$) find the short detours emanating from one of the vertices $v_p, v_{p+2\lceil \sqrt{n} \rceil}, \ldots, v_{p+2\ell\lceil \sqrt{n} \rceil}$ by running the algorithm on the graph $G^A$ obtained by adding to $G'$ the edges $(r, v_{p+2q\lceil \sqrt{n} \rceil})$ of weight $\omega(r, v_{p+2q\lceil \sqrt{n} \rceil})=q\lceil \sqrt{n} \rceil$ for every $0 \le q \le \lfloor \frac{k-p}{2\lceil \sqrt{n} \rceil} \rfloor$. Store the computed detours in the table $RD$.

This part takes $\Otilde(m\sqrt{n})$ time, as we run $\lceil \sqrt{n} \rceil$ instances of Dijkstra's algorithm whose runtime is $\Otilde(m)$.

The correctness of this algorithm for computing short detours is based on the following theorem from \cite{Roditty2005}.

\begin{theorem} [Theorem 1 in \cite{Roditty2005}] \label{appendix:thm:roditty-short-detours}
If $d_{G^A}(r, v_{2q\lceil \sqrt{n} \rceil+j}) \le (q + 1)\lceil \sqrt{n} \rceil$, where $0 \le q \le \ell$ and $0 \le j \le \lceil \sqrt{n} \rceil-1$, then
$d_{G'}(v_{2q\lceil \sqrt{n} \rceil}, v_{2q\lceil \sqrt{n} \rceil+j}) = d_{G^A}(r, v_{2q\lceil \sqrt{n} \rceil+j}) - q\lceil \sqrt{n} \rceil$. Otherwise, $d_{G'}(v_{2q\lceil \sqrt{n} \rceil}, v_{2q\lceil \sqrt{n} \rceil+j}) > \lceil \sqrt{n} \rceil$.
\end{theorem}

The basic idea of the proof of Theorem \ref{appendix:thm:roditty-short-detours} is the following. Let $0 \le q \le \ell$ and $0 \le j \le \lceil \sqrt{n} \rceil - 1$. If the distance from $r$ to $v_{2q\lceil \sqrt{n} \rceil+j}$ in $G^A$ is less than $(q+1)\lceil \sqrt{n} \rceil$ then it must be that the shortest path from $r$ to $v_{2q\lceil \sqrt{n} \rceil+j}$ starts with the edge $(r,v_{2q\lceil \sqrt{n} \rceil})$. Otherwise, if it starts with the edge $(r,v_{2z\lceil \sqrt{n} \rceil})$ for $z > q$ then the weight of this edge is at least $(z+1)\lceil \sqrt{n} \rceil$ which is already larger than $(q+1)\lceil \sqrt{n} \rceil$ contradicting the assumption that the distance from $r$ to $v_{2q\lceil \sqrt{n} \rceil+j}$ in $G^A$ is less than $(q+1)\lceil \sqrt{n} \rceil$. On the other hand, if it starts with the edge $(r,v_{2z\lceil \sqrt{n} \rceil})$ for $z < q$ then the length of any such path is at least $z \lceil \sqrt{n} \rceil + d_{G'}(v_{2z\lceil \sqrt{n} \rceil}, {2q\lceil \sqrt{n} \rceil+j}) \ge  q \lceil \sqrt{n} \rceil + d_{G}(v_{2z\lceil \sqrt{n} \rceil}, {2q\lceil \sqrt{n} \rceil+j}) = q\lceil \sqrt{n} \rceil + 2(q-z)\lceil \sqrt{n} \rceil \ge (q+1)\lceil \sqrt{n} \rceil$ where the first inequality holds since distances in $G'$ (which is obtained by removing the edges of $P_G(s,t)$ from $G$) are only larger than distances in $G$, and the last inequality holds since $z<q$.

{\bf Second part: using the table $RD$ to find replacement paths whose detour part contains at most $\lceil \sqrt{n} \rceil$ edges.}
To find the replacement path from $s$
to $t$ that avoids the edge $(v_i, v_{i+1})$ and uses a short detour, the algorithm finds
indices $i - \lceil \sqrt{n} \rceil \le a \le i$ and $i < b \le i + \lceil \sqrt{n} \rceil$ for which the expression
$d_G(s,v_{a}) + RD[a,b-a] + d_G(v_{b}, t) = a + RD[a, b - a] + (k - b)$ is minimized.
The algorithm computes it for every edge $(v_i, v_{i+1}) \in P_G(s,t)$ using a priority queue $Q$ and a sliding window approach in time $\Otilde(m\sqrt{n})$ as follows.
When looking for the shortest replacement path for the edge $(v_i, v_{i+1})$, the
priority queue $Q$ contains all pairs $(a, b)$ such that $i-\lceil \sqrt{n} \rceil \le a \le i$ and $i < b \le i+\lceil \sqrt{n} \rceil$.
The key associated with a pair $(a, b)$ is, as mentioned above, $a+RD[a, b-a]+(k-b)$.
In the
start of the iteration corresponding to the edge $(v_i, v_{i+1})$, the algorithm inserts the pairs
$(i, j)$, for $i + 1 \le j \le i + \lceil \sqrt{n} \rceil$ into $Q$, and removes from it the pairs $(j, i)$, for
$i-\lceil \sqrt{n} \rceil \le j \le i$. A find-min operation on $Q$ then returns the minimal pair $(a, b)$.

The complexity of this process is only $\Otilde(n \lceil \sqrt{n} \rceil)$: for every vertex $v_i$ (for every $0 \le i \le n$) we perform $O(\sqrt{n})$ insert operations (for all values of $j$ such that $i + 1 \le j \le i + \lceil \sqrt{n} \rceil$) which is larger than the assumed distance from $r$ to $v_{2i\lceil \sqrt{n} \rceil+j}$ in $G^A$ is at most $i\lceil \sqrt{n} \rceil+j$
$O(\sqrt{n})$ delete operations (for all values of $j$ such that $i-\lceil \sqrt{n} \rceil \le j \le i$), and a single find-min operation.
In total, we have $O(n \sqrt{n})$ operations of insert/delete/find-min which take $\Otilde(n \sqrt{n})$ time.
Thus, the total running time of the algorithm for handling short detours is $\Otilde(m \sqrt{n})$.

\subsubsection{Handling Long Detours}
To find long detours, the algorithm samples a random set $R$ as in Lemma \ref{lem:sampling-roditty} such that each vertex is sampled independently uniformly at random with probability $(4 \ln n)/\sqrt{n}$, the set $R$ has expected size of $\Otilde(\sqrt{n})$.
For every sampled vertex $r \in R$ and for every edge $e_i = (v_i, v_{i+1})$ (where $0 \le i \le k)$ we find the shortest replacement path
$P_G(s,t,e)$ which passes through $r$.

This algorithm has two steps as well.
In the first step, for every sampled vertex $r \in R$, we construct two BFS trees
from $r$, one in $G' = G \setminus E(P)$ and one in the graph obtained from $G'$
by reversing all the edge directions. This computes the distances
$d_{G'}(r, v)$ and $d_{G'}(v, r)$, for every $r \in R$ and $v \in V$.

In the second step, we run the following procedure for every sampled vertex $r \in R$.
Given $r \in R$, we find for every edge $e_i = (v_i, v_{i+1}) \in P_G(s,t)$ the shortest path from $s$ to $t$ avoiding $e_i$ which passes through $r$. To do so, we construct two priority queues
$Q_{in}[r]$ and $Q_{out}[r]$ containing indices of vertices on $P_G(s,t)$. During the computation of a
replacement path for the edge $e_i = (v_i, v_{i+1})$ we would like to run a find-min operation using priority queue $Q_{in}[r]$ to find the shortest path from $s$ to $r$ which avoids $e_i$, and we would like to run a find-min operation using priority queue $Q_{out}[r]$ to find the shortest path from $r$ to $t$ which avoids $e_i$. To do so, during the computation of a
replacement path for the edge $e_i = (v_i, v_{i+1})$ we would like to have
 $Q_{in}[r] = \{0, 1, \ldots , i\}$ and
$Q_{out}[r] = \{i + 1, \ldots , k\}$ such that an element $j \in Q_{in}[r]$ has its key in $Q_{in}[r]$ equal to
$j+d_{G'}(v_j,  r)$ and an element $j \in Q_{out}[r]$ has its key in $Q_{out}[r]$ equal to $d_{G'}(r, v_j)+(k-j)$.
Note that we have already computed $d_{G'}(v_j, r)$ in the BFS tree rooted in $r$ in the graph $G'$ with reverse edge directions and we have already computed $d_{G'}(r, v_j)$ in the BFS tree rooted in $r$ in the graph $G'$.

In order to achieve that at iteration $i$ (for $0 \le i \le k-1$) we have
$Q_{in}[r] = \{0, 1, \ldots , i\}$ and $Q_{out}[r] = \{i + 1, \ldots , k\}$ we apply the following sliding window approach.
We initiate the queues  $Q_{in}[r]$ contains only the element $0$ with key equal to $d_{G'}(v_0,  r)$ and $Q_{out}[r] = \{1, \ldots , k\}$ such that an element $j \in Q_{out}[r]$ has its key in $Q_{out}[r]$ equal to $d_{G'}(r, v_j)+(k-j)$.
Then we compute the length of the shortest path from $s$ to $t$ avoiding $e_0$ and passing through $r$ as find-min($Q_{in}[r]) + $ find-min($Q_{out}[r])$.
Next, we remove from $Q_{out}[r]$ the element $1$ and insert it to $Q_{in}[r]$ with its key equal to $1 + d_{G'}(v_1,  r)$, and compute the length of the shortest path from $s$ to $t$ avoiding $e_1$ and passing through $r$ as find-min($Q_{in}[r]) + $ find-min($Q_{out}[r])$.
In general, after finishing the $i^\text{th}$ iteration we run the $(i+1)^\text{th}$ iteration as follows.
we remove from $Q_{out}[r]$ the element $i+1$ and insert it to $Q_{in}[r]$ with its key equal to $i+1 + d_{G'}(v_{i+1},  r)$, and compute the length of the shortest path from $s$ to $t$ avoiding $e_{i+1}$ and passing through $r$ as find-min($Q_{in}[r]) + $ find-min($Q_{out}[r])$.

Finally, for every edge $e_i = (v_i, v_{i+1})$ we iterate over all vertices $r \in R$ and find the shortest path from $s$ to $t$ going through one of the vertices $R$. When $(s,t,e)$ is a long triple, there exists at least one replacement path $P_G(s,t,e)$ whose detour part contains at least $\lceil \sqrt{n} \rceil$ edges, and thus with high probability at least one of the vertices of the detour is sampled in the set $R$ (since we sample every vertex uniformly at random with probability $(4 \ln n)/\sqrt{n}$).

The total expected time of computing these distances is
$\Otilde(m\sqrt{n})$: first of all, there are $\Otilde(\sqrt{n})$ randomly chosen vertices $R$ and every BFS computation takes $O(m+n)$ time. Secondly, for every $r \in R$ we perform $O(n)$ insert, delete and find-min operations on the queues $Q_{in}[r]$ and $Q_{out}[r]$ which takes $\Otilde(n)$ time per vertex $r \in R$, and hence $\Otilde(n \sqrt{n})$ expected time. Finally for every edge $e_i$ we iterate over all the vertices $r \in R$ to find the minimum length of a shortest path from $s$ to $t$ avoiding $e_i$ which passes through one of the vertices $r \in R$. There are $O(n)$ edges $e_i \in P_G(s,t)$, and for every edge $e_i$ we iterate over $O(|R|)$ vertices which is $\Otilde(\sqrt{n})$ in expectation, and thus the total runtime of this part is $\Otilde(n\sqrt{n})$. In total we get that the algorithm takes $\Otilde(m\sqrt{n})$ time.

\subsection{The Only Randomization Used in The Replacement Paths Algorithm of Roditty and Zwick}
As mentioned above, the algorithm by Roddity and Zwick handles
separately the case that the replacement path has a short detour
containing at most $\lceil \sqrt{n} \rceil$ edges, and the case that the
replacement path has a long detour containing more than $\lceil \sqrt{n} \rceil$
edges. The first case is solved deterministically. The second
case is solved by first sampling a subset of vertices $R$ according to Lemma \ref{lem:sampling-roditty},
where each vertex is sampled uniformly independently at random with
probability $c \ln n/ \sqrt{n}$ for large enough constant $c > 0$.
Using this uniform sampling, it holds with high probability (of at least $1-n^{-c+2}$)
that for every long triple $(s,t,e)$, the detour $\Detour_{s,t,e}$ of the replacement path $P_G(s,t,e)$ contains at least one
vertex of $R$.

As the authors of \cite{Roditty2005} write, the choice of the random set $R$ is the only
randomization used by their algorithm. More precisely, the only randomization used in
the algorithm of \cite{Roditty2005} is described in the following lemma (to be self-contained, we re-write the lemma here).

\begin{lemma} [proved in \cite{Roditty2005}] \label{appendix:lemma:the-set-R-random}
Let $R \subseteq V$ be a random subset obtained by selecting each vertex, independently,
with probability $(c \ln n)/\sqrt{n}$, for some constant $c>0$. Then with high probability of at least
$1 - n^{-c+2}$, the set $R$ contains $\Otilde(\sqrt{n})$ vertices
and for every long triple $(s,t,e)$ there exists a replacement
path $P_G(s,t,e)$ whose detour part contains at least one of the
vertices of $R$.
\end{lemma}

\subsection{Derandomizing the Replacement Paths Algorithm of Roditty and Zwick - Outline}
To obtain a deterministic algorithm for the replacement paths problem and to prove Theorem \ref{thm:replacement}, we prove the following deterministic alternative to Lemma \ref{appendix:lemma:the-set-R-random} by a clever choice of the set $R$ of $\Otilde(\sqrt{n})$ vertices.

\begin{lemma} [Our derandomized version of Lemma \ref{appendix:lemma:the-set-R-random}] \label{appendix:lemma:the-set-R-deterministic}
There exists an $\Otilde(m \sqrt{n})$ time deterministic algorithm which computes
a set $R \subseteq V$ of $\Otilde(\sqrt{n})$ vertices,
such that for every long triple $(s,t,e)$ there exists a replacement
path $P_G(s,t,e)$ whose detour part contains at least one of the
vertices of $R$.
\end{lemma}

Following the above description, in order to prove Theorem \ref{thm:replacement}, that there exists an $\Otilde(m\sqrt{n})$ deterministic replacement paths algorithm, it is sufficient to prove the derandomization lemma (Lemma \ref{appendix:lemma:the-set-R-deterministic}), we do so in the following sections.
Following is an overview our approach.

We compute in $\Otilde(m\sqrt{n})$ time a set
${\cal D}_{n}$  of at most $n$ paths, each path of length exactly
$\lceil \sqrt{n} \rceil$.
The crucial part of our algorithm is in efficiently computing the set of
paths ${\cal D}_{n}$ with the following property;
for every long triple $(s,t,e)$ there exists a path $D \in
{\cal D}_{n}$ and a replacement path $P_G(s,t,e)$ such that $D$ is
contained in the detour part of $P_G(s,t,e)$. More precisely, we prove the following Lemma.

\begin{lemma} \label{appendix:thm:p-sqrt}
There exists a deterministic $\Otilde(m \sqrt{n})$ algorithm for computing a set ${\cal D}_{n}$ of
at most $n$ paths, each path of length exactly $\lceil \sqrt{n} \rceil$ with the following property;
for every long triple $(s,t,e)$ there exists a path $D \in {\cal D}_{n}$
and a replacement path $P_G(s,t,e)$ such that $D$ is the
detour part of $P_G(s,t,e)$.
\end{lemma}

After computing ${\cal D}_{n}$ we obtain the set of vertices $R$ by running the
GreedyPivotsSelection(${\cal D}_{n}$) algorithm as described in Section \ref{sec:framework}
and Section \ref{sec:framework}, and as stated in Lemma \ref{lemma:greedy}.
Given ${\cal D}_{n}$, the deterministic greedy selection algorithm GreedyPivotsSelection$(\mathcal{D}_{n})$
computes a set $R \subset V$ of $\Otilde(\sqrt{n})$ vertices in $\Otilde( n\sqrt{n})$ time with the following property;
every path $D \in \mathcal{D}_{n}$ contains at least one of the vertices of $R$.

Using Lemma \ref{appendix:thm:p-sqrt} and Lemma \ref{lemma:greedy} we can prove
the derandomization Lemma \ref{appendix:lemma:the-set-R-deterministic} and thus prove
Theorem \ref{thm:replacement}.

\begin{proof} [Proof of Theorem \ref{thm:replacement}]
According to Lemma \ref{appendix:thm:p-sqrt} we deterministically compute in $\Otilde(m \sqrt{n})$ time the set $\mathcal{D}_{n}$ of at most $n$ paths, each path of length exactly $\lceil \sqrt{n} \rceil$ with the following property; for every long triple $(s,t,e)$ there exists a path $D \in {\cal D}_{n}$
and a replacement path $P_G(s,t,e)$ such that $D$ is contained in the
detour part of $P_G(s,t,e)$.

Then, we run the greedy selection algorithm on the set of paths $\mathcal{D}_{n}$.
According to Lemma \ref{lemma:greedy}, the greedy selection algorithm takes $\Otilde(n \sqrt{n})$ time, and
computes the set $R \subset V$ of $\Otilde(\sqrt{n})$ vertices
such that every path $D \in \mathcal{D}_{n}$ contains at least one of the vertices of $R$.

We get that for every long triple $(s,t,e)$ there exists a path $D \in {\cal D}_{n}$
and a replacement path $P_G(s,t,e)$ such that $D$ is contained in the
detour part of $P_G(s,t,e)$, and the path $D \in \mathcal{D}_{n}$ contains at least one of the vertices of $R$.
Hence, for every long triple $(s,t,e)$ there exists a replacement path $P_G(s,t,e)$ whose detour part contains at least one
of the vertices of $R$.
This proves Lemma \ref{appendix:lemma:the-set-R-deterministic}.

Thus, we derandomize the randomized selection of the set of vertices $R$ in Roditty and Zwick's algorithm \cite{Roditty2005},
and as this is the only randomization used by their algorithm we obtain an $\Otilde(m\sqrt{n})$ deterministic algorithm
for the replacement paths problem in unweighted directed graphs.
\end{proof}

In the following section we prove Lemma \ref{appendix:thm:p-sqrt}.
In Section \ref{sec:reusing-common-subpaths} we already mathematically defined
the set ${\cal D}_{n}$ and in Section \ref{sec:p-sqrt-computation} we describe a deterministic
algorithm for computing $\mathcal{D}_{n}$ in $\Otilde(m\sqrt{n})$ time.

\subsection{An $\Otilde(m\sqrt{n})$ Deterministic Algorithm for Computing ${\cal D}_{n}$} \label{sec:p-sqrt-computation}

In this section we describe how to compute ${\cal D}_{n}$ in
$\Otilde(m\sqrt{n})$ time and thus proving Lemma \ref{appendix:thm:p-sqrt}.
In Section \ref{sec:rp-efficient} we presented an overview of the algorithm. We refer the reader to first read the overview in Section \ref{sec:replacement-short} before reading the following formal description of the algorithm and its analysis.

In this section we describe the algorithm more formally, and analyse its correctness and runtime.
Given the shortest path $P_G(s,t) = <s = v_0, \ldots, v_k = t>$ the
algorithm computes the weighted graph $G^w = (G,w)$ by taking the
graph $G$ and setting the weight of the edges  of $P_G(s,t)$ to be
$\epsilon$ for some small $\epsilon$ such that $0 < \epsilon < 1/n$
and the weight of all other edges to be $1$.
The algorithm then invokes Dijkstra from $s$ in $G^w$, builds the shortest paths tree
$T_s$ rooted at $s$ and sets the distances array $d[v] =
d_{G^w}(s,v)$ for every $v \in V$.
In addition, the algorithm
initializes the set $V^w_{\sqrt{n}}$ to be $V^w_{\sqrt{n}} \gets \{
v \in V \ | \ \lceil \sqrt{n} \rceil \le d[v] < \lceil \sqrt{n} \rceil + 1 \}$.
For every vertex $v\in V^w_{\sqrt{n}}$, let $\Tilde{D}^w(v)$ be the
suffix of the last $\lceil \sqrt{n} \rceil$ edges of the shortest path from $s$ to $v$ in $T_s$.
Add to the set ${\cal D}^w_{\sqrt{n}}$ (initially is set to be the empty set) the subpath $\Tilde{D}^w(v)$ for every vertex $v\in V^w_{\sqrt{n}}$.

The algorithm
then removes from $T_s$ all the vertices $v$ such that $d[v] >
\lceil \sqrt{n} \rceil + 1$ and sets $d[v] = \infty$.

Next, the algorithm operates in $|P_G(s,t)|$ iterations.
In iteration $i$ starting from $k-1$ to $0$ the algorithm does the following.
Let $T_{s,v_{i+1}}$ be the subtree of $T_s$ rooted at $v_{i+1}$.
Construct the graph $G_{s,v_{i+1}}$ as follows.
The set of vertices $V_{s,v_{i+1}}$ of $G_{s,v_{i+1}}$ is
$V_{s,v_{i+1}} = \{s\} \cup \{v\in V \mid v \in T_{s,v_{i+1}} \setminus\{v_{i+1} \} ~or~ ( \exists (v,v') \in E ~such~that~  v\in T_s\setminus\{v_{i+1} \} ~and~ v' \in T_{s,v_{i+1}})\} \setminus \{v_{i+1}\}$.
Essentially, $V_{s,v_{i+1}}$ contains all the vertices in the subtree of $T_s$ rooted in $v_{i+1}$ and all of their neighbours (except the vertex $v_{i+1}$ itself).

The set of edges of $G_{s,v_{i+1}}$ contains two types of edges.
The first type is auxiliary shortcut edges, for every vertex $v\in V_{s,v_{i+1}}$ such that $v \notin T_{s,v_{i+1}}$ add an edge $(s,v)$ with weight $d[v]$.
The second type is original edges, for every  vertex $v\in T_{s,v_{i+1}} \setminus \{v_{i+1}\} $ add all its incident edges $(v,v')$ such that $v' \in V_{s,v_{i+1}}$ with weight 1.

The algorithm removes from the tree $T_s$ the vertex $v_{i+1}$.
The algorithm then computes Dijkstra from $s$ in $G_{s,v_{i+1}}$
and for every vertex $v \in T_{s,v_{i+1}}\setminus \{v_{i+1}\}$ it sets $d[v] = d_{G_{s,v_{i+1}}}(s,v)$ and sets
the parent of $v$ in $T_s$ to be the parent of $v$ in the computed Dijkstra.
For every vertex $v \in T_{s,v_{i+1}}$ such that $d[v] > \lceil \sqrt{n} \rceil+1$ remove $v$ from  $T_s$ and set $d[v]=\infty$.
For every
vertex $v \in T_{s,v_{i+1}}$ such that $\lceil \sqrt{n} \rceil \leq d[v] < \lceil \sqrt{n} \rceil+1$
add $v$ to $V^w_{\sqrt{n}}$
and the subpath $\Tilde{D}^w(v)$ to ${\cal D}^w_{\sqrt{n}}$.

We prove the correctness and efficiency of our algorithm in the following Lemmas.

\begin{lemma} \label{lem:weighted}
Let $x$ be a vertex such that $x \in V_{\sqrt{n}}$ and let $D'(x)$ be the suffix of the last $\lceil \sqrt{n} \rceil$ edges of the shortest path from $s$ to $x$ in $G^w_{\rho(x)}$.
Then $D'(x)$ is a shortest path from $v_{\rho(x)}$ to $x$ in $G'$.
\end{lemma}

\begin{proof}
Since $x \in V_{\sqrt{n}}$ then by definition of $V_{\sqrt{n}}$ for every index $0 \le i < \rho(x)$
it holds that $d_{G'}(v_i, x) > \lceil \sqrt{n} \rceil$.
Thus, every path in $G'$ from $v_i$ to $x$ for every $i < \rho(x)$ contains
more than $\lceil \sqrt{n} \rceil$ edges, and thus any path from $s$ to $x$ in
$G^w_{\rho(x)}$ that does not pass through $v_{\rho(x)}$ has length at least
$\lceil \sqrt{n} \rceil+1$.

Therefore, the shortest path from $s$ to $x$ in $G^w_{\rho(x)}$ passes
through $v_{\rho(x)}$ and its length is $d_{G^w_{\rho(x)}}(s, v_{\rho(x)}) +
d_{G^w_{\rho(x)}}(v_{\rho(x)}, x) = \epsilon \rho(x) + d_{G'}(v_{\rho(x)},x) =  \epsilon
\rho(x) + \lceil \sqrt{n} \rceil < \lceil \sqrt{n} \rceil+1$ (where the last inequality holds as
$\epsilon < 1/n$ and $\rho(x)\leq n$).  Hence, the suffix of the last
$\lceil \sqrt{n} \rceil$ edges of the shortest path from $s$ to $x$ in $G^w_{\rho(x)}$
is a shortest path from $v_{\rho(x)}$ to $x$ in $G'$.
\end{proof}

The following lemma shows that after the $i$'th iteration $T_s$ is a shortest path tree in
$G^w_{i}$ trimmed at distance $\lceil \sqrt{n} \rceil+1$ and $d[v]$ is the distance from $s$ to $v$ in $T_s$ for every $v\in V$.

\begin{lemma}
\label{lem:TrimDijkstara}
After the $i^\text{th}$ iteration, if $d_{G^w_{i}}(s,v) < \lceil \sqrt{n} \rceil+1$ then
$v\in T_s$ and
$d[v] = d_{T_s}(s,v) = d_{G^w_{i}}(s,v)$, and otherwise $d[v] = \infty$.
\end{lemma}
\begin{proof}


We prove the claim by induction on $i$.

For $i=k$, that is, $G^w_{k} = G^w$ the claim trivially holds by the correctness of Dijkstra on $G^w$.
Assume the claim is correct for iteration $j$ such that $j>i$ and consider iteration $i$ for some $i < k$.


{\bf Consider a vertex $v$ such that $d_{G^w_{i}}(s,v) < \lceil \sqrt{n} \rceil+1$.} We also have $d_{G^w_{i+1}}(s,v) < \lceil \sqrt{n} \rceil+1$ as $G^w_{i} \subseteq G^w_{i+1}$.
By induction hypothesis we have that before iteration $i$, $d[v] = d_{T_s}(s,v)= d_{G^w_{i+1}}(s,v)$.

If $v \notin T_{s,v_{i+1}}$ then $d[v]$ and $d_{T_s}(s,v)$ do not change.
Moreover the shortest path from $s$ to $v$ in  $T_s$ before the iteration
is a shortest path in $G^w_{i+1}$, since this shortest path does not contain $v_{i+1}$ it is also a shortest path in
$G^w_{i}$. Hence, after iteration $i$ it holds that  $v \in T_s$ and $d[v] = d_{T_s}(s,v)= d_{G^w_{i}}(s,v)$.

Consider the case that $v \in T_{s,v_{i+1}} \setminus \{v_{i+1}\}$.
We prove that $d_{G_{s,v_{i+1}}}(s,v) = d_{G^w_{i}}(s,v)$.
We first prove that $d_{G_{s,v_{i+1}}}(s,v) \ge d_{G^w_{i}}(s,v)$.
By construction of $G_{s,v_{i+1}}$, every edge $(x,y)$ of $G_{s,v_{i+1}}$ is either an original edge of $G$ with weight $1$ that exists also in $G^w_{i}$, or it is an auxiliary shortcut edge $(s,y)$ such that $y \notin T_{s,v_{i+1}}$ and its weight in $G_{s,v_{i+1}}$ is defined as $\omega(s,y) = d[y]$. In the latter case, we have already proved in the previous paragraph that $d[y] = d_{G^w_{i}}(s,y)$ is the length of the shortest path from $s$ to $y$ in $G^w_i$. Hence, every $s$-to-$v$ path in $G_{s,v_{i+1}}$ is associated with an $s$-to-$v$ path in $G^w_{i}$ of the same length, and hence $d_{G_{s,v_{i+1}}}(s,v) \ge d_{G^w_{i}}(s,v)$.

We now prove that $d_{G^w_{i}}(s,v) \le d_{G_{s,v_{i+1}}}(s,v)$.
Let $P'$ be a shortest path from $s$ to $v$ in $G^w_{i}$.
Since we assumed $d_{G^w_{i}}(s,v) < \lceil \sqrt{n} \rceil+1$ and $G^w_{i} \subset G^w_{i+1}$ then $d_{G^w_{i+1}}(s,v) < \lceil \sqrt{n} \rceil+1$ and
hence by the induction hypothesis all the vertices along $P'$ are in $T_s$ at the beginning of iteration $i$.
Let $v' \in P'$ be the last vertex of $P'$ such that $v' \notin T_{s,v_{i+1}}$.
Since we assume $v \in T_{s,v_{i+1}}$, then $v' \ne v$ and all the vertices following $v'$ in $P'$ are in $T_{s,v_{i+1}}$.
Let $P_1$ be the subpath of $P'$ from $s$ to $v'$, and let $P_2$ be the subpath of $P'$ from $v'$ to $v$.
We claim that the edge $(s,v') \in G_{s,v_{i+1}}$ and its weight is $\omega(s,v') = d_{G^w_{i}}(s,v')$.
Since $v'$ is a neighbour of a vertex in $T_{s,v_{i+1}}$ ({\sl e.g.}, the vertex which follows $v'$ in $P'$ is in $T_{s,v_{i+1}}$ by definition of $v'$) then it holds that $v' \in V_{s,v_{i+1}}$.
Hence $(s,v')$ is a shortcut edge from $s$ to $v'$ whose weight equals to the weight of the shortest path $d[v']$ and as $v' \not \in T_{s, v_{i+1}}$ then we have already proved above that $v' \in T_s$ and $d[v'] = d_{T_s}(s,v')= d_{G^w_{i}}(s,v')$.
Furthermore, since all the vertices of $P'$ after $v'$ are contained in $T_{s,v_{i+1}}$ and thus in $V_{s,v_{i+1}}$,
then it follows that all the edges of $P_2$ are contained in $G_{s,v_{i+1}}$ with weight $1$ which is their original weights in $G^w_{i}$.
Hence, the path $P'' = (s, v') \cdot P_2$ is a path in $G_{s,v_{i+1}}$ whose length is $|P''| = \omega(s,v') + |P_2| = |P_1| + |P_2| = |P'|$.
Therefore, $G_{s,v_{i+1}}$ contains a $s$-to-$v$ path ({\sl e.g.}, the path $P''$) whose  length is $|P'|$.
It follows that $d_{G^w_{i}}(s,v) \le d_{G_{s,v_{i+1}}}(s,v)$.
Therefore, it holds that $d_{G_{s,v_{i+1}}}(s,v) = d_{G^w_{i}}(s,v)$ and the claim follows.

{\bf Consider the case where $d_{G^w_{i}}(s,v) > \lceil \sqrt{n} \rceil+1$.} If $d_{G^w_{i+1}}(s,v) > \lceil \sqrt{n} \rceil+1$ then the claim follows by induction hypothesis.
Otherwise it follows that $v \in T_{s,v_{i+1}}$.
It is not hard to verify that for every
vertex $u$ that belongs to $T_s$ after iteration $i$, we indeed have $d_{G_{s,v_{i+1}}}(s,u)  = d_{G^w_{i}}(s,u) \leq \lceil \sqrt{n} \rceil+1$.
Hence, since $v > \lceil \sqrt{n} \rceil+1$ we have $v \notin T_s$ and $d[v] = \infty$ after iteration $i$.

\end{proof}

The following lemmas show that every vertex $x$ belongs to at most $\lceil \sqrt{n} \rceil+1$ trees $T_{s,v_{i+1}}$. As we will later see, this will imply our desired running time.

\begin{lemma}
\label{lem:distance-increase}
Consider an iteration $i$ for some $0 \leq i \leq k-1$ and a vertex $x$ such that $x  \in T_{s,v_{i+1}}\setminus \{v_{i+1}\}$.
Then $$\lfloor d_{G^w_{i}}(s,x) \rfloor \ge \lfloor d_{G^w_{i+1}}(s,x) \rfloor +1.$$
\end{lemma}

\begin{proof}
Assume $x$ belongs to $T_{s,v_{i+1}}\setminus \{v_{i+1}\}$ for some  $0 \leq i \leq k-1$.
By Lemma \ref{lem:TrimDijkstara} after the $i$'th iteration, the shortest path between $s$ and $x$ in $T_s$ is a shortest path
between $s$ and $x$ in $G^w_i$.
Moreover, since $x$ belongs to $T_{s,v_{i+1}}\setminus \{v_{i+1}\}$
then $v_{i+1}$ is on a shortest path from $s$ to $x$ in $G^w_{i+1}$. Therefore,
$d_{G^w_{i+1}}(s,x) = \epsilon (i+1) + d_{G'}(v_{i+1}, x)$.
Observe that $\lfloor d_{G^w_{i+1}}(s,x) \rfloor$ is the length of the shortest
path from $v_{i+1}$ to $x$ in $G'$ (since $\epsilon < 1/n$).
 Assume by contradiction that $\lfloor
d_{G^w_{i}}(s,x) \rfloor \le \lfloor d_{G^w_{i+1}}(s,x) \rfloor$. Then the length of the
shortest path from $v_{i'}$ to $x$ in $G'$ for some index $i' < i+1$ is
at most the length of the path from $v_{i+1}$ to $x$ in $G'$. Then the path
from $s$ to $v_{i'}$ along $P_G(s,t)$ concatenated with the shortest path
from $v_{i'}$ to $x$  in $G'$ has length in $G^w_{i+1}$ at most $\epsilon i' + \lfloor
d_{G^w_{i+1}}(v_{i'},x) \rfloor$ and hence
it is shorter than $d_{G^w_{i+1}}(s,x) = \epsilon (i+1) + d_{G'}(v_{i+1}, x)$
which is a contradiction since $d_{G^w_{i+1}}(s,x)$ is the length of
the shortest path from $s$ to $x$ in $G^w_{i+1}$.
\end{proof}

By Lemmas \ref{lem:TrimDijkstara} and \ref{lem:distance-increase} and the fact that the algorithm trims the tree $T_s$ at distance $\lceil \sqrt{n} \rceil+1$ we get the following.

\begin{lemma}
\label{lem:num-of-Dijkstra}
Every vertex $x$ belongs to at most $\lceil \sqrt{n} \rceil+1$ subtrees $T_{s,v_{i+1}}$
for $0 \leq i \leq k-1$.
\end{lemma}

\begin{proof}
Assume $x$ belongs to the trees $T_{s,v_{i_1}}...T_{s,v_{i_r}}$ for $i_1< i_2<...<i_r$.
We will show that $r \leq \lceil \sqrt{n} \rceil+1$, which implies the lemma.

By Lemma \ref{lem:TrimDijkstara} as long as $x \in T_s$ for some iteration $j$ we have that after iteration $j$ it holds that
$d[v] = d_{G^w_{j}}(s,v)$. By Lemma \ref{lem:distance-increase} we have
$\lfloor d_{G^w_{i_{j}}}(s,x) \rfloor +1 \leq
\lfloor d_{G^w_{i_j -1} }(s,x) \rfloor \leq
\lfloor d_{G^w_{i_{j-1}}}(s,x) \rfloor$.


Since the algorithm only maintains vertices $v$ whose
distance $d[v]$ is less than $\lceil \sqrt{n} \rceil+1$ then after $x$
participates in $\lceil \sqrt{n} \rceil+1$ subtrees $T_{s,v_{i+1}}$ the algorithm removes it from $T_s$ and therefore $r \leq \lceil \sqrt{n} \rceil+1$ as required.
\end{proof}

\begin{lemma}
The total running time of the algorithm is $\tilde{O}(m\sqrt{n})$.
\end{lemma}
\begin{proof}
The dominant part of the running time of the algorithm is the computations of Dijkstra.
The first Dijkstra computation is on the graph $G^w$ and it takes $O(m+n\log{n})$ time.

We claim that the computation of iteration $i$ takes $O(\sum_{v\in T_{s,v_{i+1}}}{\deg(v)} \log{n})$, where $\deg(v)$ is the degree of $v$ in $G$.
To see this, note that both the number of nodes and the number of edges in $G_{s,v_{i+1}}$ is bounded by
$O(\sum_{v\in T_{s,v_{i+1}}}{\deg(v)})$.
By Lemma \ref{lem:num-of-Dijkstra} every node $v$ belongs to at most $\lceil \sqrt{n} \rceil+1$ trees
$T_{s,v_{i+1}}$.
It is not hard to see now that the lemma follows.
\end{proof}

The following Lemma proves Lemma \ref{thm:p-sqrt}.

\begin{lemma}
$V^w_{\sqrt{n}} = V_{\sqrt{n}}$ and ${\cal D}^w_{\sqrt{n}}$ can be chosen as the set ${\cal D}_{n}$ according to Definition
\ref{def:p-sqrtn}.
\end{lemma}

\begin{proof}
{\bf We first prove that $V^w_{\sqrt{n}} \subseteq V_{\sqrt{n}}$.}
Let $x \in V^w_{\sqrt{n}}$. As $x \in V^w_{\sqrt{n}}$ then there exists an iteration $i'$ such that after iteration $i'$,
$\lceil \sqrt{n} \rceil \leq d[x]  < \lceil \sqrt{n} \rceil+1$.
Consider the tree $T_s$ after iteration $i'$.
By Lemma \ref{lem:TrimDijkstara} after iteration $i'$, the path from $s$ to $x$ in $T_s$ is of length $d[x]$ and is a shortest path in $G^w_{i'}$.
Let $P$ be the shortest path from $s$ to $x$ in $T_s$.

Let $i$ be the maximal index such that $v_{i}$ is on the shortest path $P$.
As $P$ does not contain any vertex $v_{j}$ such that $i+1 \leq j$ then $P$ is also a shortest path in $G^w_{i}$.
As $P$ is of length between $\lceil \sqrt{n} \rceil$ and $\lceil \sqrt{n} \rceil+1$ we prove that $i = \rho(x)$.
We need to prove that $d_{G'}(v_i, x) = \lceil \sqrt{n} \rceil$ and that the distance from $v_j$ for every $j\leq i$ to $x$ in $G'$ is more than $\lceil \sqrt{n} \rceil$.

We first prove that $d_{G'}(v_i, x) = \lceil \sqrt{n} \rceil$.
Since $P$ is a shortest $s$-to-$x$ path in $T_s$ and $i$ be the maximal index such that $v_{i}$ is on the shortest path $P$ then $P$ is composed of the path $<v_0, \ldots, v_i>$ followed by a shortest path from $v_i$ to $x$ in $G'$. That is $P = <v_0, \ldots, v_i> \cdot P_{G'}(v_i, x)$ (where $P_{G'}(v_i, x)$ is a shortest path from $v_i$ to $x$ in $G'$) . We have that $\lceil \sqrt{n} \rceil \le |P| =  |<v_0, \ldots, v_i>| + |P_{G'}(v_i, x)| = \epsilon i + d_{G'}(v_i, x) < \lceil \sqrt{n} \rceil+1$ and $|<v_0, \ldots, v_i>| = i \epsilon < 1$ (as $\epsilon < 1/n$). Since all the edges of $P_{G'}(v_i, x)$ have weight $1$, we get that $|P_{G'}(v_i, x)| = d_{G'}(v_i, x) = \lfloor |P| \rfloor = \lceil \sqrt{n} \rceil$.

Next, we prove that the distance from $v_j$ for every $j < i$ to $x$ in $G'$ is more than $\lceil \sqrt{n} \rceil$.
Indeed, assume by contradiction there exists an index $j < i$ such that $d_{G'}(v_j, x) \le \lceil \sqrt{n} \rceil$.
Then the path $P' = <v_0, \ldots, v_j> \circ P_{G'}(v_j, x)$ (where $P_{G'}(v_j, x)$ is a shortest path from $v_j$ to $x$ in $G'$) has length $|P'| = \epsilon j + d_{G'}(v_j, x) \le  \epsilon j + \lceil \sqrt{n} \rceil < \epsilon i + \lceil \sqrt{n} \rceil = |P|$. We get that $P'$ which is a $s$-to-$x$ path in $G^w_{i}$ is shorter than $P$ which is a shortest $s$-to-$x$ path in $G^w_{i}$, which is a contradiction.
We proved that $d_{G'}(v_i, x) = \lceil \sqrt{n} \rceil$ and $d_{G'}(v_j, x) > \lceil \sqrt{n} \rceil$ for every $j\leq i$ and hence $i = \rho(x)$.
By Lemma \ref{lem:weighted} it follows that $\Tilde{D}^w(x)$ (which is the subpath containing the last $\lceil \sqrt{n} \rceil$ edges of $P$) is a shortest path from $v_{\rho(x)}$ to $x$ in $G'$.
Note also that the algorithm adds to ${\cal D}^w_{\sqrt{n}}$ the subpath $\Tilde{D}^w(v)$.

{\bf We now prove that $V_{\sqrt{n}} \subseteq V^w_{\sqrt{n}}$.}
Let $x \in V_{\sqrt{n}}$ then there exists an index $0 \le \rho(x) \le k$
such that $d_{G'}(v_{\rho(x)}, x) = \lceil \sqrt{n} \rceil$ and for all $0 \le i < \rho(x)$
it holds that $d_{G'}(v_i, x) \ge \lceil \sqrt{n} \rceil+1$.
Consider the tree $T_s$ at the end of iteration $\rho(x)$.

We first prove that $d_{G^w_{\rho(x)}}(s, x) < \lceil \sqrt{n} \rceil+1$ and hence by Lemma
\ref{lem:TrimDijkstara} it follows that $x \in T_s$.
Let $P$ be the following path from $s$ to $x$.
$P = <v_0, \ldots, v_{\rho(x)}> \circ P_{G'}(v_{\rho(x)}, x)$, that is, the path composed of the first
$\rho(x)$ edges from $s$ to $v_{\rho(x)}$ along $P_G(s,t)$ followed by a shortest path from $v_{\rho(x)}$ to $x$ in $G'$.
Since $|P| = \epsilon \rho(x) + d_{G'}(v_{\rho(x)},x) = \epsilon \rho(x) + \lceil \sqrt{n} \rceil < \lceil \sqrt{n} \rceil+1$
and $P$ is a path in $G^w_{\rho(x)}$ it follows that $d_{G^w_{\rho(x)}}(s,x) \le |P| < \lceil \sqrt{n} \rceil+1$.

Next, we prove that $d_{G^w_{\rho(x)}}(s, x) \ge \lceil \sqrt{n} \rceil$.
Let $P' = P_{G^w_{\rho(x)}}(s, x)$ be a shortest path from $s$ to $x$ in $G^w_{\rho(x)}$.
Let $i$ be the maximal index such that $v_i \in  P'$, then the subpath of $P'$ from $v_i$ to $x$ is a shortest path in $G'$ and its length $d_{G'}(v_i, x)$. Since $i \le \rho(x)$ (as $P'$ is a path in $G^w_{\rho(x)}$) then by Definition \ref{def:x-tag} it follows that $d_{G'}(v_i, x) \ge \lceil \sqrt{n} \rceil$. Therefore, $d_{G^w_{\rho(x)}}(s, x) = |P'| \ge d_{G'}(v_i, x) \ge \lceil \sqrt{n} \rceil$.

By Lemma \ref{lem:TrimDijkstara} after the $\rho(x)$ iteration,
$d[v] = d_{G^w_{\rho(x)}}(s,x)$. Hence, by the end of the $\rho(x)$ iteration
$\lceil \sqrt{n} \rceil \leq d[v] < \lceil \sqrt{n} \rceil+1$.
Therefore, by construction $x \in V^w_{\sqrt{n}}$.
Let $P$ be the shortest $s$-to-$x$ path in $T_s$ after the  $\rho(x)$ iteration.
By Lemma \ref{lem:TrimDijkstara} it holds that $P$ is a shortest $s$-to-$x$ path in $G^w_{\rho(x)}$, and
by Lemma \ref{lem:weighted} it follows that $\Tilde{D}^w(x)$ (which is the subpath containing the last $\lceil \sqrt{n} \rceil$ edges of $P$) is a shortest path from $v_{\rho(x)}$ to $x$ in $G'$.
Note that the algorithm adds to ${\cal D}^w_{\sqrt{n}}$ the subpath $\Tilde{D}^w(v)$.

In the proof above, we have also proved that every vertex $x \in V_{\sqrt{n}}$ we add a single path $D(x)$ to ${\cal D}^w_{\sqrt{n}}$.
The path $D(x)$ is obtained by the algorithm by taking the last $\lceil \sqrt{n} \rceil$ edges of a shortest path from $s$ to $x$ in the graph $G^w_{\rho(x)}$, and we have already proved that it is a shortest path from $\rho(x)$ to $x$ in $G'$ whose length is $\lceil \sqrt{n} \rceil$. This proves that ${\cal D}^w_{\sqrt{n}}$ can be used as the set of paths ${\cal D}_{n}$ according to Definition \ref{def:p-sqrtn}.
\end{proof}

\subsection{An Alternative $\Otilde(m\sqrt{n})$ Deterministic Algorithm for Computing ${\cal D}_{n}$} \label{sec:p-sqrt-computation2}
We shortly describe an alternative $\Otilde(m\sqrt{n})$
deterministic algorithm for computing ${\cal D}_{n}$. The algorithm
of Roddity and Zwick \cite{Roditty2005} for handling short detours
constructs $2\lceil \sqrt{n} \rceil$ auxiliary graphs $G^A_0,
\ldots, G^A_{2\lceil \sqrt{n} \rceil-1}$ (see Figure \ref{fig:the-graph-G-A} as an illustration of
the graph $G^A = G^A_0$). For every $0 \le z \le
2\lceil \sqrt{n} \rceil-1$, the auxiliary graph $G^A_z$ is obtained
by adding a new source vertex $r_z$ to $G'$ and an edge $(r_z,
v_{z+2q \lceil \sqrt{n} \rceil})$ of weight $\omega(r_z, v_{z+2q}) = q \lceil \sqrt{n} \rceil$ for every integer $0
\le q \le \frac{\sqrt{n}}{2}$. The weight of all the edges $E \setminus E(P)$ is
set to $1$. Then run Dijkstra's algorithm from $r_z$
in $G^A_z$ that computes a shortest paths tree $T_z$.

We claim that given the shortest paths trees $T_0, \ldots, T_{2\lceil \sqrt{n} \rceil-1}$, the
following algorithm computes the set of paths ${\cal D}_n$. For
every $v \in V$ the algorithm computes the minimum index $0 \le i
\le k$ such that at least one of the shortest paths trees $T_0, \ldots, T_{2\lceil \sqrt{n} \rceil-1}$
contains a path from $v_i$ to $v$ of length $\lceil \sqrt{n} \rceil$, and let $P'(v)$ be this path from $v_i$ to
$v$. If no such index $0 \le i \le k$ exists, then set $P'(v) =
\emptyset$. For every vertex $v \in V$ finding this minimum index
$i$ takes $O(\sqrt{n})$ time as there are $O(\sqrt{n})$ shortest
paths trees to check, and in every shortest paths tree $T_z$ we only need to
check the first edge $(r_z, v_j)$ of the path from $r_z$ to $v$.
So, given the shortest
paths trees $T_0, \ldots, T_{2\lceil \sqrt{n} \rceil-1}$, computing the paths ${\cal D}_n = \{
P'(v) \ | \ v \in V\}$ takes $O(n\sqrt{n})$ time.

\begin{lemma} \label{lem:alg-alternative}
Let $x \in V_{\sqrt{n}}$
and let $z := (\rho(x) \mod 2\lceil \sqrt{n} \rceil)$.
Then the shortest paths tree $T_z$ contains a shortest path in $G'$ from $v_{\rho(x)}$ to $x$
of length $\lceil \sqrt{n} \rceil$.
\end{lemma}

\begin{proof}
Since $x \in V_{\sqrt{n}}$
then $d_{G'}(v_{\rho(x)}, x) = \lceil \sqrt{n} \rceil$
and thus the shortest path in $G'$ from $v_{\rho(x)}$ to $x$ contains exactly $\lceil \sqrt{n} \rceil$ edges.
We prove that every shortest path in $G^A_z$ from $r_z$ to $x$ must start with the edge $(r_z, v_{\rho(x)})$.

Let $P' = (r_z, v_{\rho(x)}) \circ P_{G'}(v_{\rho(x)}, x)$.
Assume there exists a path $P_1$ in $G^A_z$ from $r_z$ to $x$ that starts with the edge $(r_z, v')$
such that $v' = v_{\rho(x)-i\cdot  2\lceil \sqrt{n} \rceil}$ for some integer $i>0$.
Then $\omega(P_1) = \omega(r_z, v') +
d_{G'}(v', x) = \omega(r_z, v_{\rho(x)}) - i \cdot  \lceil \sqrt{n} \rceil
+ d_{G'}(v', x) \ge \omega(r_z, v_{\rho(x)}) - i \cdot  \lceil \sqrt{n} \rceil + (2i+1)\cdot  \lceil \sqrt{n} \rceil = \omega(r_z, v_{\rho(x)}) + (i+1)\cdot  \lceil \sqrt{n} \rceil >
\omega(r_z, v_{\rho(x)}) +  \lceil \sqrt{n} \rceil =
\omega(P')$.
Then $P'$ is shorter than $P_1$ and thus $P_1$ is not a shortest path in $G^A_z$.

Assume there exists a path $P_2$ in $G^A_z$ from $r_z$ to $x$ that starts with the edge $(r_z, v')$
such that $v' = v_{\rho(x) + i\cdot  2\lceil \sqrt{n} \rceil}$ for some integer $i>0$.
Then $\omega(P_2) = \omega(r_z, v') +
d_{G'}(v', x) = \omega(r_z, v_{\rho(x)}) + i \cdot  \lceil \sqrt{n} \rceil
+ d_{G'}(v', x) > \omega(r_z, v_{\rho(x)}) +  \lceil \sqrt{n} \rceil = \omega(P')$.
Then $P'$ is shorter than $P_2$ and thus $P_2$ is not a shortest path in $G^A_z$.

It follows that every shortest path from $r_z$ to $x$ in $G^A_z$ must start with the edge $(r_z, v_{\rho(x)})$.
In particular, the shortest paths tree $T_z$ contains a shortest path in $G'$ from $v_{\rho(x)}$ to $x$
of length $\lceil \sqrt{n} \rceil$.
\end{proof}

Finally, let ${\cal D}_{n} := \{P'(v) \ | \ v \in V\}$ be the set of paths computed by the above algorithm.
Then ${\cal D}_{n}$ is a set of $O(n)$ paths, each path contains exactly
$\lceil \sqrt{n} \rceil$ edges, and it follows from Lemma \ref{lem:alg-alternative} and
Definition \ref{def:p-sqrtn} that a subset of the paths ${\cal D}_{n}$
satisfies the conditions of Definition \ref{def:p-sqrtn}.
Hence, it is sufficient to use the greedy algorithm GreedyPivotsSelection to hit the set of paths ${\cal D}_{n}$.

\section{Deterministic $f$-Sensitivity Distance Oracles} \label{sec:dso1}
As explained in the introduction, an
$f$-Sensitivity Distance Oracle gets as an input a graph $G$ and a
parameter $f$, preprocesses it into a data-structure, such that given
a query $(s,t,F)$ with $s,t \in V, F \subseteq E \cup V, |F| \le f$
one may efficiently compute the distance $d_{G\setminus F}(s,t)$.

In this section we derandomize the result of Weimann and Yuster
\cite{WY13} for real edge weights. They presented an $f$-sensitivity
distance oracle whose preprocessing time is
$\Otilde(mn^{1+\epsilon})$ (which is larger than the time it takes
to compute APSP which is $O(mn)$ by a factor of $n^\epsilon$) and
whose query time is subquadratic $\Otilde(n^{2-2\epsilon/f})$. More
precisely, we prove the following theorem which obtains
deterministically the same preprocessing and query time bounds as
the randomized $f$-sensitivity distance oracle in \cite{WY13} for real edge weights.

\begin{theorem}
Let $G$ be a weighted directed graph, and let $f \ge 1$ be a
parameter. One can deterministically construct an $f$-sensitivity
distance oracle in $\Otilde(mn^{1+\epsilon})$ time, such that given
a query $(s,t,F)$ with $F \subset V \cup E$ and $|F| \le f$ the deterministic
query algorithm for computing $d_{G\setminus F}(s,t)$ takes subquadratic $\Otilde(n^{2-2\epsilon/f})$ time.
\end{theorem}

The basic idea in derandomizing the real weighted $f$-sensitivity
distance oracles of Weimann and Yuster \cite{WY13} is to use a
variant of the fault tolerant trees $FT^{L,f}(s,t)$ described in
Appendix A in \cite{ChCoFiKa17} to find short replacement paths, and
then use the greedy algorithm from Section \ref{sec:framework} and
Lemma \ref{lemma:greedy} for derandomizing the random selection of
the pivots, and finally continue with the algorithm of \cite{WY13}
for stitching short segments to obtain the long replacement paths.
This overview is made clear in the description below.

According to Section \ref{sec:assumption-unique}, we will assume WLOG the following holds.
\begin{itemize}
\item  {\bf Unique shortest paths assumption: } we assume that all the shortest paths are unique.
\item {\bf Non-negative weights assumption: } we assume that edge weights are non-negative, so that we can run Dijkstra.
\end{itemize}

%
%
%



{\bf Outline.} Let $s,t \in V$ be vertices and let $f, L\ge 1$ be
integer parameters. In Section \ref{sec:ft-trees} we described the
trees $FT^{L,f}(s,t)$ which are a variant of the trees that appear
in Appendix A of \cite{ChCoFiKa17}. In Section
\ref{sec:dynamic-programming} we described how to construct the
trees $FT^{L,f}(s,t)$ in $\Otilde(mn^{1+\epsilon+\epsilon/f} +
n^{2+\epsilon+2\epsilon/f})$ time.

In Section \ref{sec:ft-trees-appendix} we prove Lemma
\ref{lemma:query}, that the algorithm described in Section
\ref{sec:ft-trees} computes the distance $d^L_G(s,t,F)$ in $O(f^2
\log L)$ time.
In Section \ref{sec:tree-to-dso} we describe how to
use the trees $FT^{L,f}(s,t)$ in order to construct an
$f$-sensitivity distance oracle. In Section \ref{sec:positive} we
reduce the construction time of the trees $FT^{L,f}(s,t)$ to
$\Otilde(mn^{1+\epsilon})$. In Section \ref{sec:improved-greedy} we
reduce the runtime of the GreedyPivotsSelection algorithm from
$\Otilde(n^{2+\epsilon+\epsilon/f})$ to $\Otilde(n^{2+\epsilon})$.
In Section \ref{sec:assumption-unique} we justify our assumptions of
non-negative edge weights and unique shortest paths.

\subsection{Proof of Lemma \ref{lemma:query}}\label{sec:ft-trees-appendix}

\begin{proof}[Proof of Lemma \ref{lemma:query}]
We first prove correctness, that the query procedure outputs $d^L_G(s,t,F)$.
The query procedure outputs a distance $d^L_G(s,t, \{a_1, \ldots, a_i\})$ such that $\{a_1, a_2, \ldots, a_i \} \subseteq F$ and
$P^L_G(s,t, \{a_1, \ldots, a_i\}) \cap F = \emptyset$ (note that this includes the case that $d^L_G(s,t, \{a_1, \ldots, a_i\}) = \infty$ when there is no path from $s$ to $t$ in $G \setminus \{a_1, a_2, \ldots, a_i \}$ that contains at most $L$ edges).
The distance $d^L_G(s,t, \{a_1, \ldots, a_i\})$ is the minimum length of an $s$-to-$t$ path that contains at most $L$ edges in the graph $G \setminus \{a_1, a_2, \ldots, a_i \}$.
On the one hand, since $\{a_1, a_2, \ldots, a_i \} \subseteq F$ and distances may only increase as we delete more and more vertices and edges, we obtain that $d^L_G(s,t, \{a_1, \ldots, a_i\}) \le d^L_G(s,t, F)$.
On the other hand, since $P^L_G(s,t, \{a_1, \ldots, a_i\}) \cap F = \emptyset$ then $P^L_G(s,t, \{a_1, \ldots, a_i\})$ is a path in the graph $G \setminus F$ that contains at most $L$ edges, and hence its length $d^L_G(s,t, \{a_1, \ldots, a_i\})$ is at least the length of the shortest $s$-to-$t$ path in $G \setminus F$ that contains at most $L$ edges which is $d^L_G(s,t, F)$, so we get $d^L_G(s,t, \{a_1, \ldots, a_i\}) \ge d^L_G(s,t,F)$.
It follows that the output of the query is $d^L_G(s,t, \{a_1, \ldots, a_i\}) = d^L_G(s,t,F)$.

We analyze the runtime of the query. The runtime of the query is $O(f^2\log L)$ as we
advance along a root to leaf path in $FT^{L,f}(s,t)$ (whose length is at
most $f$) and in each node $FT^{L,f}(s,t, a_1, \ldots, a_i)$ of the tree
$FT^{L,f}(s,t)$ we make $O(f)$ queries $a_{i+1} \in F$ to $BST^L(s,t, a_1, \ldots, a_i)$ which take $O(\log L)$ as we search in a binary
search tree with $L$ elements. So query time is the multiplication of the following terms:
\begin{itemize}
\item $f$  --- length of root-to-leaf path in $FT^{L,f}(s,t)$
\item $f$  --- number of elements $a_{i+1} \in F$ to check in the node $FT^{L,f}(s,t, a_1, \ldots, a_i)$ whether or not $a_{i+1} \in P^L_G(s,t, \{a_1, \ldots, a_i\})$
\item $O(\log L)$ --- time to check for a single element $a_{i+1} \in F$ whether or not $a_{i+1} \in P^L_G(s,t, \{a_1, \ldots, a_i\})$ by searching $a_{i+1}$ in $BST^L(s,t, a_1, \ldots, a_i)$ which is a binary search tree containing $L$ elements.
\end{itemize}
\end{proof}

\subsection{Deterministic $f$-Sensitivity Distance Oracles with $\Otilde(mn^{1+\epsilon+\epsilon/f} + n^{2 + \epsilon + 2\epsilon/f})$ Preprocessing Time}
\label{sec:tree-to-dso} In this section we describe how to plug-in
the trees $FT^{L,f}(s,t)$ from Section \ref{sec:ft-trees} in the
$f$-sensitivity distance oracles of Weimann and Yuster \cite{WY13}.

Let us first recall how the $f$-sensitivity distance oracle of Weimann and Yuster \cite{WY13} works. The following Lemma is proven in \cite{WY13}.

\begin{lemma} [Theorem 1.1 in \cite{WY13}] \label{lem:randomized-f-dso}
Given a directed graph $G$ with real
positive edge weights, an integer parameter $1 \le f \le \epsilon \log n/ \log \log n$ and a real
parameter $0 < \epsilon < 1$, there exists a randomized
$f$-sensitivity distance oracle whose construction is randomized and
takes $\Otilde(mn^{1+\epsilon})$ time. Given a query $(s,t,F)$ where
$F \subset V \cup E$ with $|F| \le f$, the data-structure answers
the query by computing w.h.p. $d_G(s,t,F)$ in $\Otilde(n^{2-2\epsilon/f})$ time.
\end{lemma}

\begin{proof}
Use $\alpha = 1-\epsilon$ in the construction of Weimann and Yuster \cite{WY13}.
In order for this result to be self-contained, we briefly explain the preprocessing and query procedures of Weimann and Yuster \cite{WY13}.

{\bf Preprocessing: }
\begin{enumerate}
\item Randomly generate graphs $G_1, \ldots, G_r$ (with $r = \Otilde(n^\epsilon)$) where every graph is independently obtained by removing each edge with probability $1/n^{\epsilon/f}$.
Compute APSP on each of the graphs $\{G_1, \ldots, G_r\}$.
It is proven in \cite{WY13} that with high probability for every set $F \subseteq E \cup V$ with $|F| \le f$ and for every $s-t$ shortest path $P_G(s,t,F)$ that contains at most $n^{\epsilon/f}$ edges in the graph $G\setminus F$, there exists at least one graph $G_i \in \{G_1, \ldots, G_r\}$ that excludes $F$ and contains $P_G(s,t,F)$.

\item Sample a random set $B$ of pivots, where every vertex is taken with probability $\frac{6f \ln n}{n^{\epsilon/f}}$.

\end{enumerate}

{\bf Query: }
Given a query $(s,t,F)$ build the dense graph $H$ (denoted by $G^S$ in \cite{WY13}) whose vertices are $B \cup \{s, t\}$ as follows.
First, find all the graphs that exclude $F$, $\mathcal{G}_F = \{ G_i \ | \ 1 \le i \le r, F \cap G_i = \emptyset \}$.
Then, for every $u,v \in B \cup \{s,t\}$ add the edge $(u,v)$ to $H$ and set is weight to be the minimum length of the shortest path from $u$ to $v$ in all the graphs $\mathcal{G}_F$. If there is no path from $u$ to $v$ in any of the graphs $\mathcal{G}_F$ then set $\omega_H(u,v) = \infty$.
Finally, run Dijkstra from $s$ in the graph $H$ and output $d_G(s,t,F) = d_{H}(s,t)$.
\end{proof}


We derandomize Lemma \ref{lem:randomized-f-dso} as follows.

\begin{lemma} \label{lem:f-dso-mn2}
Given a directed graph $G$ with real positive edge weights, an
integer parameter $1 \le f \le \epsilon \log n/ \log \log n$ and $0 < \epsilon < 1$, there exists a deterministic
$f$-sensitivity distance oracle whose construction is deterministic
and takes $\Otilde(mn^{1+\epsilon+\epsilon/f} + n^{2+\epsilon+2\epsilon/f})$
time. Given a query $(s,t,F)$ where $F \subset V \cup E$ with $|F|
\le f$, the data-structure answers deterministically the query by
computing $d_G(s,t,F)$ in $\Otilde(n^{2-2\epsilon/f})$ time.
\end{lemma}

To prove Lemma \ref{lem:f-dso-mn2} we describe how to use the fault-tolerant trees $\{ FT^{L,f}(s,t) \}_{s,t \in V}$ to construct the deterministic $f$-sensitivity distance oracle.

{\bf Preprocessing: }
\begin{enumerate}
\item {\bf Compute the trees $FT^{L,f}(u,v)$ }. Deterministically construct the fault-tolerant trees $FT^{L,f}(u,v)$ for every $u,v \in V$ as in Lemma \ref{lemma:ft-tree-eps-plus}.

\item {\bf Compute the set of vertices $B \subseteq V$ }.
Let $\mathcal{P}^{L,f}$ be the set of all paths in all the nodes of all the trees $\{FT^{L,f}(u,v)\}_{u,v \in V}$ that contain at least $n^{\epsilon/f}/2$
edges.
Use the greedy algorithm as in Lemma
\ref{lemma:greedy}, to find deterministically a set of pivots $B$,
such that for every $P \in \mathcal{P}^{L,f}$ it holds that $B \cap V(P)
\ne \emptyset$.
\end{enumerate}

 {\bf Query: } Given a query $(s,t,F)$ build the complete graph
$H$ (also referred to as the dense graph) whose vertices are $B \cup \{s, t\}$ as follows. For every $u,v
\in B \cup \{s,t\}$  query the tree $FT^{L,f}(u,v)$ with $(u,v,F)$
according to the query procedure described in Section
\ref{sec:ft-trees} and set the weight of the edge $(u,v)$ in $H$ to the distance computed $d^{L}_G(u,v,F)$ ({\sl i.e.}, set $\omega_H(u,v) = d^L_G(u,v,F)$).
Finally, run Dijkstra from $s$ in the graph $H$ and output $d_{H}(s,t)$ as an estimate of the distance $d_G(s,t,F)$.

\begin{proof}[Proof of Lemma \ref{lem:f-dso-mn2}]
We prove the correctness of the DSO and then analyse its preprocessing and query time.

{\bf Proof of correctness}.
We prove that $d_G(s,t,F) = d_H(s,t)$.
Since $d_H(s,t)$ is the length of some path from $s$ to $t$ in $G \setminus F$ then $d_H(s,t) \ge d_G(s,t,F)$.
Next we prove that $d_G(s,t,F) \ge d_H(s,t)$.

Let $P_G(s,t,F)$ be the shortest path from $s$ to $t$ in $G\setminus F$.
We prove that the set of vertices $B$ hits every subpath of $P_G(s,t,F)$ that contains exactly $n^{\epsilon/f}/2$ edges.
To see that, let $u,v$ be two vertices such that $u$ appears before $v$ along $P_G(s,t,F)$ and $P_G(s,t,F)[u..v]$ contains $n^{\epsilon/f}/2$ edges. It follows that the shortest path from $u$ to $v$ in $G \setminus F$ contains less than $L = n^{\epsilon/f}$ edges.
According to the assumption of unique shortest paths (as in Section \ref{sec:unique}) it follows that querying the tree $FT^{L,f}(u,v)$ with $(u,v,F)$ finds a path $P^L_G(u,v,F)$, and since the shortest path from $u$ to $v$ in $G \setminus F$ contains less than $L = n^{\epsilon/f}$ edges then $P^L_G(u,v,F) = P_G(u,v,F) = P_G(s,t,F)[u..v]$ is the shortest path from $u$ to $v$ in $G \setminus F$, and it contains exactly $n^{\epsilon/f}/2$ edges.
Therefore, $P^L_G(u,v,F)$ is a path in a node of the tree $FT^{L,f}(u,v)$ and thus $P^L_G(u,v,F) \in \mathcal{P}^{L,f}$. Hence, when computing the hitting set $B$ that hits all the paths of $\mathcal{P}^{L,f}$ the algorithm obtains a set $B$, and in particular it holds that $B \cap P^L_G(u,v,F) \ne \emptyset$ and thus
$B$ hits the path $P_G(s,t,F)[u..v] = P_G(u,v,F)$.
This prove that $B$ hits every subpath of $P_G(s,t,F)$ that contains exactly $n^{\epsilon/f}/2$ edges.

Let $s = v_1, v_2, \ldots, v_k = t$ be all the vertices of $H = B \cup \{s, t \}$ that appear along $P_G(s,t,F)$ sorted according to the order of their appearance along the path $P_G(s,t,F)$.
As $B$ hits every subpath of $P_G(s,t,F)$ that contains exactly $n^{\epsilon/f}/2$ edges, it follows that for every $0 \le i < k$ the subpath of $P_G(s,t,F)$ from $v_i$ to $v_{i+1}$ contains at most $n^{\epsilon/f}$ edges.
Therefore, the shortest path from $v_i$ to $v_{i+1}$ in $G \setminus F$ contains at most $L = n^{\epsilon/f}$ edges and hence $d^L_G(v_i, v_{i+1}, F) = d_G(v_i, v_{i+1}, F)$.
By Lemma \ref{lemma:query} it holds that querying the tree $FT^{L,f}(v_i, v_{i+1})$ according to the query procedure described in Section
\ref{sec:ft-trees} computes the distance $d^L_G(v_i, v_{i+1}, F) = d_G(v_i, v_{i+1}, F)$.
Hence $\omega_H( <v_1, \ldots, v_k> ) = \omega_H(v_1, v_2) + \ldots + \omega_H(v_{k-1}, v_k) = d^L_G(v_1, v_2, F) + \ldots + d^L_G(v_{k-1}, v_k, F) =
d_G(v_1, v_2, F) + \ldots + d_G(v_{k-1}, v_k, F) = d_G(s,t,F)$ where the last equality holds as
$s = v_1, v_2, \ldots, v_k = t$ are the vertices of $H$ that appear along $P_G(s,t,F)$ sorted according to the order of their appearance along the path $P_G(s,t,F)$.

It follows that the path $<v_1, \ldots, v_k>$ is an $s$-to-$t$ path in $H$ whose weight in $H$ is $d_G(s,t,F)$ and hence
the shortest $s$-to-$t$ path in $H$ has length at most $\omega_H(<v_1, \ldots, v_k>) = d_G(s,t,F)$. Therefore $d_H(s,t) \le \omega_H(<v_1, \ldots, v_k>) = d_G(s,t,F)$ and since we already proved that $d_H(s,t) \ge d_G(s,t,F)$ it follows that $d_H(s,t) = d_G(s,t,F)$.

{\bf Analysing the preprocessing time}.
Next we analyse preprocessing time of the DSO.
Constructing the trees $FT^{L,f}(u,v)$ for every $u,v \in V$ as in Lemma \ref{lemma:ft-tree-eps-plus} takes
$\Otilde(mn^{1+\epsilon+\epsilon/f} + n^{2+\epsilon+2\epsilon/f})$ time.
We analyse the time it takes to compute the set of vertices $B \subseteq V$.
Let $\mathcal{P}^{L,f}$ be the set of all paths in all the nodes of all the trees $\{FT^{L,f}(u,v)\}_{u,v \in V}$ that contain at least $n^{\epsilon/f}/2$
edges.
Observe that $\mathcal{P}^{L,f}$ contains at most $O(n^{2+\epsilon})$
paths, as the number of nodes in each of the $n^2$ trees
$FT^{L,f}(u,v)$ is $O(n^\epsilon)$ and every such node contains a
single path of $G$.
Thus, using the greedy algorithm as in Lemma
\ref{lemma:greedy} finds deterministically in $\Otilde(n^{2+\epsilon +
\epsilon/f})$ time a set of pivots $B$,
such that $|B| = \Otilde(n^{1-\epsilon/f})$ and $B$ hits all the paths in $\mathcal{P}^{L,f}$.
Thus, the total preprocessing time of the DSO is $\Otilde(mn^{1+\epsilon+\epsilon/f} + n^{2+\epsilon+2\epsilon/f})$ time.

{\bf Analysing the query time}.
Next we analyse the query time of the DSO.
During query, the algorithm constructs the graph $H$. During the construction of $H$ the algorithm runs
$\Otilde(n^{2-2\epsilon/f})$ queries using the trees
$\{FT^{L,f}(u,v)\}_{u,v \in V}$, each query is answered in  $O(f^2
\log n)=\Otilde(1)$ time since we assumed $f \le \log n / \log \log
n$. Thus, constructing the graph $H$ takes $\Otilde(n^{2-2\epsilon/f})$ time, and also running Dijkstra from $s$ in the graph $H$ takes
$\Otilde(n^{2-2\epsilon/f})$ time.
Therefore, the total query time is $\Otilde(n^{2-2\epsilon/f})$ time.
\end{proof}

\subsection{Deterministic $f$-Sensitivity Distance Oracles with $\Otilde(mn^{1+\epsilon})$ Preprocessing Time}
\label{sec:positive}
In this section we reduce the preprocessing time of constructing the
$f$-sensitivity distance oracle described in Lemma
\ref{lem:f-dso-mn2} from $\Otilde(mn^{1+\epsilon+\epsilon/f} +
n^{2+\epsilon+2\epsilon/f})$ to match the preprocessing time of
\cite{WY13} which is $\Otilde(mn^{1+\epsilon})$, while keeping the
same query time of $\Otilde(n^{2-2\epsilon/f})$.

We improve the preprocessing time in two ways: improving the
construction time of the trees $FT^{L,f}(u,v)$ and reducing the
runtime of the greedy selection algorithm by considering a smaller
set of paths $\mathcal{P}^{L,f}$. In Section \ref{sec:ft-trees-n-plus-eps}
we reduce the time it takes to construct the trees $\{FT^{L,f}(u,v)
\}_{u,v \in V}$ from $\Otilde(mn^{1+\epsilon+\epsilon/f} +
n^{2+\epsilon+2\epsilon/f})$ to $\Otilde(mn^{1+\epsilon})$. In
Section \ref{sec:improved-greedy} we show that the runtime of the
greedy selection algorithm of the pivots $B$ can be reduced from
$\Otilde(n^{2+\epsilon+\epsilon/f})$ to $\Otilde(n^{2+\epsilon})$
which is negligible compared to $\Otilde(mn^{1+\epsilon})$. This
will give us the desired $\Otilde(mn^{1+\epsilon})$ preprocessing
time.

\subsubsection{Building the Trees $FT^{L,f}(u,v)$ in $\Otilde(mn^{1+\epsilon})$ Time} \label{sec:ft-trees-n-plus-eps}

In this section we describe how to construct the trees
$FT^{L,f}(u,v)$ in $\Otilde(mn^{1+\epsilon})$ time. We first define
the node $FT^{L',f}(u,v, a_1, \ldots, a_i)$ of the tree
$FT^{L',f}(u,v)$ as follows.

\begin{definition}
Let $1 \le L' \le L,  0 \le i \le f$ and $u,v \in V$.
Assume $a_1 \in P^{L'}(u,v), a_2 \in P^{L'}(u,v, \{a_1\}),  \ldots, a_i \in P^{L'}(u,v, \{a_1, \ldots, a_{i-1}\})$.
We define the node $FT^{L',f}(u,v, a_1, \ldots, a_i)$ of depth $i$ in the tree $FT^{L',f}(u,v)$ as the node we reach if we query the tree $FT^{L',f}(u,v)$ with $F = \{a_1, \ldots, a_i \}$ according to the query procedure described in Section \ref{sec:ft-trees}. If $i=f$ then $FT^{L',f}(u,v, a_1, \ldots, a_f)$ is a leaf node of $FT^{L',f}(u,v)$ of depth $f$. We slightly abuse notation and use $FT^{L',f}(u,v, a_1, \ldots, a_i)$ to both denote the node $FT^{L',f}(u,v, a_1, \ldots, a_i)$ and the subtree of $FT^{L',f}(u,v)$ rooted in $FT^{L',f}(u,v, a_1, \ldots, a_i)$.

%
\end{definition}

Recall that in Section \ref{sec:ft-trees} we described how to build
the trees $FT^{L,f}(u,v)$ in
$\Otilde(mn^{1+\epsilon+\epsilon/f}+n^{2+\epsilon+2\epsilon/f})$
time. More generally, we built the trees $\{ FT^{L',f}(u,v) \}_{1
\le L' \le L, u,v \in V}$ in $\Otilde(mn L^{f+1} + n^2 L^{f+2})$
time, where the construction time consists of the following two
terms $\Otilde(mn L^{f+1})$ and $\Otilde(n^2 L^{f+2})$. The first
term $\Otilde(mn L^{f+1})$ is the time it takes to solve the dynamic
programming Equation \ref{eq:dynamic-programming} in all the nodes
of all the trees $\{ FT^{L',f}(u,v) \}_{u,v \in V, 1 \le L' \le L}$.
In Section \ref{sec:dijsktra-last-layer}  we reduce the runtime of
this part to $\Otilde(mn L^{f})$ by applying Dijkstra on auxiliary
graphs $H_{F,t}$ we define later rather than computing the
dynamic-programming in the last layer of the trees.

The second term $\Otilde(n^2 L^{f+2})$ is the time it takes to reconstruct all the paths
$P^{L'}(s,t, \{a_1, \ldots, a_i\})$ that are explicitly stored  in all the nodes
of all the trees $\{ FT^{L',f}(u,v) \}_{u,v \in V, 1 \le L' \le L}$.
It takes $\Otilde(n^2 L^{f+2})$ time since the number of nodes in all the trees $\{ FT^{L',f}(u,v) \}_{u,v \in V, 1 \le L' \le L}$
 is $\Otilde(n^2 L^{f+1})$, and it takes $O(L)$ time to reconstruct each path (which contains at most $L$ edges).
In Section \ref{sec:improving-paths-reconstruction} we reduce this term to $\Otilde(n^2 L^f) = \Otilde(n^{2+\epsilon})$ by not reconstructing the paths
in the leaves of the trees of depth $f$.



\subsubsection{The improved algorithm for constructing the trees $FT^{L,f}(u,v)$}\label{sec:improving-paths-reconstruction}
We describe the algorithm with the improved construction time.
First, the algorithm constructs the trees $\{ FT^{L',f}(u,v) \}_{1 \le L' \le L, u,v \in V}$ up to level $f-1$, {\sl i.e.}, constructing the trees $\{ FT^{L',f-1}(u,v) \}_{1 \le L' \le L, u,v \in V}$, without constructing the nodes in level $f$.
Note that using the analysis above this takes $\Otilde(mn L^{(f-1)+1}) = \Otilde(mn L^f) = \Otilde(mn^{1+\epsilon})$ time.

We are left with explaining how to accelerate the construction of the last layer (the layer of depth $f$) of the trees $\{ FT^{L,f}(u,v) \}_{u,v \in V}$.
The algorithm reconstructs the paths in level $f-1$ only for the trees $\{ FT^{L,f}(u,v) \}_{u,v \in V}$.
That is, for every leaf node $FT^{L,f-1}(u,v, a_1, \ldots, a_{f-1})$ of $FT^{L,f-1}(u,v)$ we can reconstruct the path
$P^L_G(u,v, \{a_1, \ldots, a_{f-1}\})$ in $O(L)$ time by following the parent pointers $\textrm{parent}^{L}(u,v, F)$ as computed in Equation \ref{eq:parent-pointer}.
Reconstructing these $\Otilde(n^2 L^{f-1})$ paths $P^L_G(u,v, \{a_1, \ldots, a_{f-1}\})$ take $\Otilde(n^2 L^f) = \Otilde(n^{2+\epsilon})$ time.
Then for every vertex or edge $a_f \in P^L_G(u,v, \{a_1, \ldots, a_{f-1}\})$ we need to construct the leaf node $FT^{L,f}(u,v, a_1, \ldots, a_{f-1}, a_f)$.
We describe the construction of the leaves $FT^{L,f}(u,v, a_1, \ldots, a_{f-1}, a_f)$ in the following paragraphs.

Let ${\cal F}_{u,v} = \{ \{a_1, \ldots, a_f \} \  | \text{ the node } FT^{L,f}(u,v, a_1, \ldots, a_f) \text{ is a leaf node of } FT^{L,f}(u,v) \}$.
Equivalently, ${\cal F}_{u,v} = \{ \{a_1, \ldots, a_f \} \  |  \ a_1 \in P^{L}(u,v), a_2 \in P^{L}(u,v, \{a_1\}),  \ldots, a_f \in P^{L}(u,v, \{a_1, \ldots, a_{f-1}\}) \}$.
In the remaining of this section we describe how to compute the distances $d^L_G(u,v,F)$ and the parent pointer
parent$^L(u,v,F)$ for every $u,v \in V,  F \in {\cal F}_{u,v}$ in total $\Otilde(mn L^f)$ time.

We first explain why it is not possible to use Equation \ref{eq:dynamic-programming} to compute the distance $d^L_G(u,v, \{a_1, \ldots, a_f\})$.
Let $F = \{a_1, \ldots, a_f \} \in {\cal F}_{u,v}$.  According to Equation \ref{eq:dynamic-programming},
$d^{L}_G(u,v, F) = \min_{z} \{ \omega(u,z) + d^{L-1}(z,v,F) \ | \  (u,z) \in E \ \ AND \ \ u,z,(u,z) \not \in F \}$, where the distance $d^{L-1}(z,v,F)$
need to be obtained by querying the tree $FT^{L-1,f}(z,v)$ using the query $F$.
But querying the tree $FT^{L-1,f}(z,v)$ with the set $F$ might reach a leaf node $FT^{L-1,f}(z,v, a_1, \ldots, a_f)$,
and since we did not construct yet the last layer (of level $f$) of $FT^{L-1,f}(z,v)$
then we do not have the distance $d^{L-1}(z,v, F)$ computed yet. This breaks the dynamic programming of Equation \ref{eq:dynamic-programming}.
To overcome this difficulty, we run Dijkstra in auxiliary graphs $H_{F,t}$.

\subsubsection{The auxiliary graphs $H_{F,t}$} \label{sec:dijsktra-last-layer}
Let $t \in V$ be a fixed vertex, we define ${\cal F}_{t} = \cup_{s \in V} {\cal F}_{s,t}$.
For every $F \in {\cal F}_t$ we build the graph $H_{F,t} = (V_{F,t}, E_{F,t})$ as follows.
Let $V'_{F,t} \subseteq V$ be the set of all vertices $s \in V$
such that $F \in {\cal F}_{s,t}$.


Note that we can easily compute the sets $V'_{F,t}$ for all the vertices $t\in V$
in $\Otilde(n^2 L^f)$ time using the following procedure.
Initialize an empty hash table $h$.
For every $s \in V, F \in {\cal F}_{s,t}$:
\begin{itemize}
\item Check if $h$ contains $(F,t)$. If $(F,t) \not \in h$ then set $h[F,t]$ to be an empty list of vertices (which will eventually represent $V'_{F,t}$).
\item Append $s$ to the end of the list $h[F,t]$.
\end{itemize}
After scanning all the sets $F \in {\cal F}_{s,t}$ for every $s,t \in V$, we get that the vertices $V'_{F,t}$
are listed in $h[F,t]$.
The runtime of this procedure is $\Otilde(n^2 L^f)$ for all the vertices $t\in V$
as the number of sets $F$ in $\bigcup_{s,t \in V} {\cal F}_{s,t}$ is at most the number of
leaves in all the trees $\{ FT^{L,f}(s,t) \}_{s,t \in V}$, which is $\Otilde(n^2 L^f)$.

Let $V_{F,t}$ be the set of vertices $V'_{F,t}$ and their neighbours, {\sl i.e.}, $V_{F,t} = V'_{F,t} \cup N(V'_{F,t})$.

Let $s \in V_{F,t}$ and let $u$ be a neighbour of $s$ in $G$,
{\sl i.e.}, $(s,u) \in E$, such that $s,u$ has not failed ({\sl i.e.}, $s,u, (s,u) \not \in F$).
If $u \in V_{F,t}$ then we add the edge $(s,u)$ to $H_{F,t}$ with its weight $\omega(s,u)$.
Otherwise, $u \not \in V_{F,t}$, this means that if we query the tree $FT^{L,f}(u,t)$ with the set $F$ then the query ends in an internal node of $FT^{L,f}(u,t)$ and not in a leaf of depth $f$. Let $FT^{L,f}(u,t, a_1, \ldots, a_i)$ be the internal node we reach at the end of the query such that $\{a_1, \ldots, a_i\} \subsetneq F$. Then the path $P^L_G(u,t, \{a_1, \ldots, a_i\})$ does not contain $F$ (as otherwise we would have not finished the query of $F$ in the tree $FT^{L,f}(u,t)$ at this node), and hence we already computed $d^L_G(u,t,F) = d^L_G(u,t, \{a_1, \ldots, a_i\})$ in the node $FT^{L,f}(u,t, a_1, \ldots, a_i)$.
We add an edge $(u,t)$ to $H_{F,t}$ and assign to it the weight $d^L_G(u,t,F)$, and we refer to these edges as {\sl shortcuts}.

Finally, we compute Dijkstra from $t$ in the graph $H_{F,t}$ with reverse edge directions.
This gives us the distances $d_{H_{F,t}}(s,t)$ for every $s \in V_{F,t}$.
We claim that $d_{H_{F,t}}(s,t) \le d^L_G(s,t,F)$ and that computing all the Dijkstra's in all the graphs $H_{F,t}$ takes $\Otilde(mn^{1+\epsilon})$ time.

\begin{lemma} \label{lem:H-dijkstra-correctness}
Let $s \in V_{F,t}$, then $d_G(s,t,F) \le d_{H_{F,t}}(s,t) \le d^L_G(s,t,F)$.
\end{lemma}

\begin{proof}
We first prove that $d_G(s,t,F) \le d_{H_{F,t}}(s,t)$.
If  $d_{H_{F,t}}(s,t) < \infty$ then the shortest path from $s$ to $t$ in $H_{F,t}$ is composed of two types of edges:
\begin{itemize}
\item An edge $(u,v)$ such that $u, v,(u,v) \not \in F$ whose weight equals to $\omega(u,v)$ (its weight in the graph $G$).
In this case, the edge $(u,v)$ also exists in the graph $G \setminus F$ with the same weight.
\item A ``shortcut'' edge $(u,t)$ whose weight is $d^L_G(u,t,F)$.
In this case, there is a path from $u$ to $t$ in the graph $G \setminus F$ whose weight is $d^L_G(u,t,F)$.
\end{itemize}

In both cases we get that for every edge $(u,v)$ in the graph $H_{F,t}$ there exists a $u$-to-$v$ path in the graph $G \setminus F$ whose weight in $G \setminus F$ equals to the weight of the edge $(u,v)$ in $H_{F,t}$. Hence,  $d_G(s,t,F) \le d_{H_{F,t}}(s,t)$.

We now prove that $d_{H_{F,t}}(s,t) \le d^L_G(s,t,F)$.
Let $P^L_G(s,t,F) = <v_0, \ldots, v_k>$ be a shortest path from $s$ to $t$ in $G \setminus F$ on at most $L$ edges ({\sl i.e.}, $k \le L$).
We prove that there exists an $s$-to-$t$ path $P$ in $H_{F,t}$ whose weight $\omega_{H_{F,t}}(P) \le \omega_G(P^L_G(s,t,F))$ and since $d_{H_{F,t}}(s,t) \le \omega_{H_{F,t}}(P)$
it follows that $d_{H_{F,t}}(s,t) \le \omega_{H_{F,t}}(P) \le \omega_G(P^L_G(s,t,F)) = d^L_G(s,t,F)$.
In order to construct the path $P$ in $H_{F,t}$, note that for every edge $(v_i, v_{i+1})$ of $P^L_G(s,t,F)$ (for every $0 \le i < k$) it holds that either $(v_i, v_{i+1}) \in G \setminus F$ exists in $H_{F,t}$ with the same weight as its weight in $G$, or there exists an edge $(v_i, t)$ in $H_{F,t}$ whose weight is $d^L_G(v_i, t, F)$.
Let $0 \le \ell \le k$ be the maximum index such that for every $0 \le i < \ell$ the edge $(v_i, v_{i+1}) \in H_{F,t}$.
We define the path $P := <v_0, \ldots, v_\ell> \circ (v_\ell, t)$, then $P$ is a path in $H_{F,t}$ and its weight is
$\omega_{H_{F,t}}(P) = \omega_{H_{F,t}}(<v_0, \ldots, v_\ell>) + \omega_{H_{F,t}}(v_\ell, t) =
\omega_G(<v_0, \ldots, v_\ell>) + d^L_G(v_\ell, t, F)$ where the last inequality holds as
every $0 \le i < \ell$ the edge $(v_i, v_{i+1})$ exists in $H_{F,t}$ with the same weight as its weight in $G$
and the edge $(v_\ell, t)$ exists in $H_{F,t}$ and its weight in $H_{F,t}$ is $d^L_G(v_i, t, F)$.
It follows that $\omega_{H_{F,t}}(P) = \omega_G(<v_0, \ldots, v_\ell>) + d^L_G(v_\ell, t, F) = d^{\ell}_G(s,v_\ell) + d^L_G(v_\ell, t, F) \le d^L_G(s,t)$
where the last equality holds by the triangle inequality and the fact that $<v_0, \ldots, v_\ell> \subseteq P^L_G(s,t,F)$.
Hence, the shortest path from $s$ to $t$ in $H_{F,t}$ has weight  at most $d^L_G(s,t,F)$, and thus $d_{H_{F,t}}(s,t) \le d^L_G(s,t,F)$.
\end{proof}

\begin{lemma} \label{lem:H-dijkstra-time}
Computing Dijkstra's algorithm in all the graphs $H_{F,t}$ takes $\Otilde(mn L^f) = \Otilde(mn^{1+\epsilon})$ time.
\end{lemma}

\begin{proof}
The runtime of Dijkstra in the graph $H_{F,t}$ is $\Otilde(\Sigma_{\{s \in V'_{F,t} \}} \text{deg}(s))$, as $O(\Sigma_{\{s \in V'_{F,t} \}} \text{deg}(s))$ is a bound on the number of edges and vertices in $H_{F,t}$.

It follows that the runtime of running all Dijkstra algorithms is $\Otilde(\Sigma_{t \in V} \Sigma_{F \in {\cal F}_t} \Sigma_{\{s \in V'_{F,t} \}} \text{deg}(s))$.
Note that a vertex $s \in V'_{F,t}$ iff $F = \{a_1, \ldots, a_f \}$ and $FT^{L,f}(s,t, a_1, \ldots, a_f)$ is a leaf node of $FT^{L,f}(s,t)$ at depth $f$.
Hence, for a fixed vertex $t \in V$ it holds that $\Sigma_{F \in {\cal F}_t} \Sigma_{\{s \in V'_{F,t} \}} \text{deg}(s)$ is the sum of $\text{deg}(s)$ for every leaf node of $FT^{L,f}(s,t)$.
As $FT^{L,f}(s,t)$ contains $\Otilde(L^f)$ leaves, then
$\Sigma_{F \in {\cal F}_t} \Sigma_{\{s \in V'_{F,t} \}} \text{deg}(s) = \Otilde(L^f \cdot \Sigma_{s \in V} \text{deg}(s)) = \Otilde(m L^f)$.
Therefore,
\newline \noindent $\Otilde(\Sigma_{t \in V} \Sigma_{F \in {\cal F}_t} \Sigma_{\{s \in V'_{F,t} \}} \text{deg}(s)) = \Otilde(nm L^f) = \Otilde(mn^{1+\epsilon})$, where the last equality holds as $L = n^{\epsilon/f}$.
\end{proof}

\subsubsection{Reducing the Runtime of the Greedy Selection Algorithm}
\label{sec:improved-greedy}

We have $O(n^2)$ trees $\{FT^{L,f}(s,t)\}_{s,t \in V}$, every tree
contains $O(n^\epsilon)$ nodes, and every node $FT^{L,f}(s,t, a_1,
\ldots, a_i)$ contains a path $P^L_G(s,t, \{a_1, \ldots, a_i\})$ with
at most $L = n^{\epsilon/f}$ edges. In the greedy algorithm we want to
hit all of these paths that contain at least $n^{\epsilon/f}/2$
edges and at most $n^{\epsilon/f}$ edges. In total there might be
$O(n^{2+\epsilon})$ such paths $\{P_G(s,t, a_1, \ldots, a_i)\}$, each
path contains at least $n^{\epsilon/f}/2$ edges and at most
$n^{\epsilon/f}$ edges, and thus according to Lemma
\ref{lemma:greedy} finding a set of vertices $R$ of size
$\Otilde(n^{1-\epsilon/f})$ which hits all these paths takes
$\Otilde(n^{2+\epsilon+\epsilon/f})$ time.

Let $R_{<f}$ be the hitting set of vertices obtained by the greedy
algorithm which hits all the paths ${\cal P}_{<f} = \{P_G(s,t, \{a_1,
\ldots, a_i\}) | 1 \le i < f\}$ that contains at least
$n^{\epsilon/f}/4$ edges and at most $n^{\epsilon/f}$ edges, these
are paths that appear in the internal nodes of the trees
$FT^{L,f}(s,t)$ (which are not in the last layer of the trees).
Since there are only $O(n^{2+\epsilon-\epsilon/f})$ such paths
${\cal P}_{<f}$, each path contains at least $n^{\epsilon/f}/8$
edges and at most $n^{\epsilon/f}$ edges, and thus according to
Lemma \ref{lemma:greedy} finding a set of vertices $R_{<f}$ of size
$\Otilde(n^{1-\epsilon/f})$ which hits all of these paths takes
$\Otilde(n^{2+\epsilon})$ time.

We define $P_{\text{remaining}}$ to be the subset of paths $\{P_G(s,t, \{a_1, \ldots, a_f\})\}$ for which the following conditions hold:
\begin{itemize}
\item $P_G(s,t, \{a_1, \ldots, a_f\})$ is a path stored in a some leaf node ($FT^{L,f}(s,t, a_1, \ldots, a_f)$) of depth $f$ in at least one of the trees $FT^{L,f}(s,t)$.
\item $P_G(s,t, \{a_1, \ldots, a_f\})$ contains between $n^{\epsilon/f}/2$ to $n^{\epsilon/f}$ edges.
\item $P_G(s,t, \{a_1, \ldots, a_f\})$ does not contain any of the vertices $R_{<f}$.
\end{itemize}

Following we describe how to compute in $\Otilde(mn^{1+\epsilon})$ time a set ${\cal P}_f$ of $\Otilde(n^{2+\epsilon-\epsilon/f})$ paths, each path contains at least $n^{\epsilon/f}/8$ edges, such that if we hit all the paths $P_{f}$ then we also hit every path of $P_{\text{remaining}}$.

In Lemma \ref{lem:H-dijkstra-time} we run Dijkstra in the graph $H_{F,t}$ and computed shortest paths to $t$,
let $T_{F,t}$ be the shortest paths tree rooted in $t$ in the graph $H_{F,t}$.
Let $X_{F,t}$ be all the vertices $x \in V_{F,t}$ in the tree $T_{F,t}$ at depth $n^{\epsilon/f}/8$ ({\sl i.e.}, the number of edges from the root of $T_{F,t}$ to $x$ is $n^{\epsilon/f}/8$) such that there exists at least one vertex $y \in V_{F,t}$ which is a descendent of $x$ in $T_{F,t}$ and $y$ is at depth $n^{\epsilon/f}/4$  in $T_{F,t}$.
Let $P_{F,t}$ be the set of paths in the tree $T_{F,t}$ from every vertex $x \in X_{F,t}$ to the root $t$ where a shortcut edge $(u,t)$ is replaced with the subpath $P^L_G(u,t,F)$, so that every path in $P_{F,t}$ is a valid path in $G \setminus F$.
Finally, let ${\cal P}_f = \bigcup_{F,t} P_{F,t}$.

We claim that ${\cal P}_f$ is a set of $\Otilde(n^{2+\epsilon-\epsilon/f})$ paths,
each path contains at least $n^{\epsilon/f}/8$ edges, such that if we hit all the paths ${\cal P}_f$ then we also hit all the paths $P_{\text{remaining}}$.
We first need the following lemma.

%
%
%
%

\begin{lemma} \label{lem:H-bound-vertices}
The total number of vertices in all the graphs $H_{F,t}$ is $\Sigma_{F,t} |V_{F,t}| = \Otilde(n^{2+\epsilon})$.
\end{lemma}

\begin{proof}
Since every vertex of $V_{F,t}$ is either a vertex of $V'_{F,t}$ or a neighbour of such a vertex, then it holds that $\Sigma_{F,t} |V_{F,t}| \le  \Sigma_{F,t} \Sigma_{\{x \in V'_{F,t} \} \text{deg}(x)}$.

Note that a vertex $x \in V'_{F,t}$ iff querying the tree
$FT^{L,f}(x,t)$ with $F$ results in reaching a leaf at depth $f$ of
the tree $FT^{L,f}(x,t)$. Hence, for a fixed vertex $x \in V$, the
sum $(\Sigma_{F,t \ | \ x \in V'_{F,t}} \text{deg}(x))$ is bounded
by the number of nodes in the last layer of all the trees $\{
FT^{L,f}(x,t) \}_{t \in V}$ multiplied by $\text{deg}(x)$. Since the
last layer of every tree $FT^{L,f}(x,t)$ contains $n^{\epsilon}$
nodes, and for every vertex $x \in V$ there are $n$ trees $\{
FT^{L,f}(x,t) \}_{t \in V}$, then we get a bound $\Otilde(\Sigma_{x
\in V} n^{1+\epsilon} \text{deg}(x)) = \Otilde(mn^{1+\epsilon})$ on
the number of vertices in all the graphs $H_{F,t}$.
\end{proof}

\begin{lemma}
${\cal P}_f$ is a set of $\Otilde(n^{2+\epsilon-\epsilon/f})$ paths, each path contains at least $n^{\epsilon/f}/8$ edges, such that if we hit all the paths $P_{f}$ then we also hit all the paths $P_{\text{remaining}}$.
The runtime to compute ${\cal P}_f$ is $\Otilde(mn^{1+\epsilon})$.
\end{lemma}

\begin{proof}
Let $P_G(s,t, \{a_1, \ldots, a_f\}) = P_G(s,t,F) \in
P_{\text{remaining}}$ be the path stored in the node $FT^{L,f}(s,t,
a_1, \ldots, a_f)$ such that $\{a_1, \ldots, a_f \} = F$. Denote by
$P_G(s,t,F) = <v_1, \ldots, v_r>$.

Since $P_G(s,t,F) \in P_{\text{remaining}}$ then $P_G(s,t,F)$ contains
between $n^{\epsilon/f}/2$ to $n^{\epsilon/f}$ edges and it is not
hit by $R_{<f}$. Let $1 \le i < r- n^{\epsilon/f}/4-1$, then $v_i$
is a vertex on the path $P_G(s,t,F)$ which is not among the last
$n^{\epsilon/f}/4$ vertices of the path. Then $P_G(v_i,t,F)$ is a
subpath of $P_G(s,t,F)$ (since shortest paths are unique).
Furthermore, since $R_{<f}$ (which hits all the paths which contain
at least $n^{\epsilon/f}/4$ vertices in all the non-leaf nodes of
all the trees) does not hit $P_G(v_i,t,F)$ then it follows that
$P_G(v_i,t,F)$ is stored in a leaf node $FT^{L,f}(v_i,t,F)$ of the
tree $FT^{L,f}(v_i,t)$. A similar argument shows that
$P_G(v_{i+1},t,F)$ is stored in a leaf node $FT^{L,f}(v_{i+1},t,F)$ of
the tree $FT^{L,f}(v_{i+1},t)$. It follows that $v_i, v_{i+1} \in
V_{F,t}$ and $(v_i, v_{i+1}) \in H_{F,t}$. Therefore, $P_H(s,t)$
contains at least all the edges $(v_i, v_{i+1})$ for every $1 \le i
< r- n^{\epsilon/f}/4-1$, and since $r$ is the number of vertices of
$P_G(s,t,F)$ which contains at least $n^{\epsilon/f}/2$ vertices then
$P_H(s,t)$ contains at least $n^{\epsilon/f}/4$ vertices.

Since $P_H(s,t)$ is a path in $H_{F,t}$ containing at least $n^{\epsilon/f}/4$ vertices then it holds for the $n^{\epsilon/f}/8$-th vertex $x$ from the end of $P_H(s,t)$
that $x \in X_{F,t}$. Therefore, the subpath $P_H(x,t)$ of $P_H(s,t)$ from $x$ to $t$ which contains $n^{\epsilon/f}/8$ edges is contained in ${\cal P}_f$. Hence, if we hit ${\cal P}_f$ we also hit $P_H(x,t)$ and therefore we also hit $P_G(s,t,F)$.  This proves that hitting all the paths of ${\cal P}_f$ also hits all the paths of $P_{\text{remaining}}$.

Next, we prove that ${\cal P}_f$ contains $\Otilde(n^{2+\epsilon-\epsilon/f})$ paths.
We have already proved in \ref{lem:H-dijkstra-time} that the number of vertices in all the graphs $H_{F,t}$ is $\Otilde(n^{2+\epsilon})$.
Recall that ${\cal P}_f = \bigcup_{F,t} \{P_H(x,t) \ | \ x \in X_{F,t} \}$.
Furthermore, for every vertex $x \in X_{F,t}$ there exists at least $n^{\epsilon/f}/8$ unique vertices in the subtree of $x$. To see this, recall that by definition if the vertex $x \in X_{F,t}$ there exists at least one vertex $y \in V_{F,t}$ which is a descendent of $x$ in $T_{F,t}$ and $y$ is at depth $n^{\epsilon/f}/4$ in $T_{F,t}$. Thus, the set of vertices of $T_{F,t}$ from $x$ to $y$ contains at least $n^{\epsilon/f}/8$ vertices which belong to the subpath of $T_{F,t}$ rooted at $x$.

Therefore, $|{\cal P}_f| = \sum_{F,t} |X_{F,t}| = \Otilde(n^{2+\epsilon} / (n^{\epsilon/f}/8) = \Otilde(n^{2+\epsilon-\epsilon/f})$.

Finally, the run time to compute ${\cal P}_f$ is $\Otilde(mn^{1+\epsilon})$ as it is dominated by the Dijkstra algorithms in the graphs $H_{F,t}$ whose runtime is $\Otilde(mn^{1+\epsilon})$ according to Lemma \ref{lem:H-dijkstra-time}, and every path in ${\cal P}_f$ contains at least $n^{\epsilon/f}/8$ edges by definition.
\end{proof}

After computing ${\cal P}_f$ in $\Otilde(n^{2+\epsilon})$ time, we
run the greedy selection algorithm from Lemma \ref{lemma:greedy} on
the set of paths ${\cal P}_f$ in $\Otilde(n^{2+\epsilon})$ time
(note that the bound on the runtime follows as $|{\cal P}_f| =
\Otilde(n^{2+\epsilon - \epsilon/f})$) to obtain a set $R_f$ of
$\Otilde(n^{1-\epsilon/f})$ vertices that hit all the paths $P_{f}$
and thus they also hit all the paths $P_{\text{remaining}}$. Let $R
= R_{<f} \cup R_f$. So in total this takes $\Otilde(n^{2+\epsilon})$
time to find the set $R$ of $\Otilde(n^{1-\epsilon/f})$ vertices
that hit all the paths $\{P_G(s,t, \{a_1, \ldots, a_i\})\}$ in all the
nodes of all the trees $FT^{L,f}(s,t)$ which contain at least
$n^{\epsilon/f}/2$ edges and at most $n^{\epsilon/f}$ edges.

\begin{corollary}
One can find deterministically in $\Otilde(n^{2+\epsilon})$ time a
set $R$ of $\Otilde(n^{1-\epsilon/f})$ vertices that hit all the
paths $\{P_G(s,t, \{a_1, \ldots, a_i\})\}$ in all the nodes of all the
trees $FT^{L,f}(s,t)$ which contain at least $n^{\epsilon/f}/2$
edges and at most $n^{\epsilon/f}$ edges.
\end{corollary}

\subsection{Assumptions} \label{sec:assumption-unique}
In the algorithms we described for the case of directed graphs with real edge weights for constructing and querying the DSO we made two assumptions:
\begin{itemize}
\item We assumed all edge weights are non-negative, so that we can run Dijkstra algorithm.
\item We assumed all the shortest paths: $P_G(s,t), P_G(s,t,F), P^L_G(s,t), P^L_G(s,t,F)$ are unique.
\end{itemize}

In this section we justify these two assumptions.

\subsubsection{Handling Negative Weights} \label{sec:negative}
In the description above we assumed that edge weights are non-negative.
In this section we describe how to reduce the problem of general edge weights to non-negative edge weights.

We handle it similarly to Weimann and Yuster \cite{WY13} by the well known method of feasible price functions in order to transform the negative edge weights to be nonnegative in the graph $G$ as the first step of the preprocessing algorithm. For a directed graph $G = (V,E)$ with (possibly negative) edge weights $\omega(\cdot)$, a price function is a function $\phi$ from the vertices of $G$ to the reals. For an edge $(u,v)$, its reduced weight with respect to $\phi$ is $\omega_\phi(u,v) = \omega(u,v) + \phi(u) - \phi(v)$. A price function $\phi$ is feasible if $\omega_\phi(u,v) \ge 0$ for all edges $(u,v) \in E$. The reason feasible price functions are used in the computation of shortest paths is that for any two vertices $s,t \in V$, for any $s$-to-$t$ path $P$, $\omega_\phi(P) = \omega(P) + \phi(s) - \phi(t)$. This shows that an $s$-to-$t$ path is shortest with respect to $\omega_\phi(\cdot)$ iff it is shortest with respect to $\omega(\cdot)$. Moreover, the $s$-to-$t$ distance with respect to the original weights $\omega(\cdot)$ can be recovered by adding $\phi(t)-\phi(s)$ to the $s$-to-$t$ distance with respect to $\omega_\phi(\cdot)$.

The most natural feasible price function comes from single-source
distances. Let $s$ be a new vertex added to $G$ with an edge from
$s$ to every other vertex of $G$ having weight $0$. Let $d(v)$
denote the distance from $s$ to vertex $v \in G$. Then for every
edge $(u,v) \in E$, we have that $d(v) \le d(u) + \omega(u,v)$, so
$\omega_d(u,v) \ge 0$ and thus $d( \cdot )$ is feasible. This means
that knowing $d(\cdot)$, we can now use Dijkstra's SSSP algorithm on
$G$ (with reduced weights) from any source we choose and obtain the
SSSP with respect to the original $G$.

Therefore, we first compute $\phi = d(\cdot)$ in the original graph,
store $\phi$ and change the weights of every edge $(u,v)$ to
$\omega_d(u,v)$ which are non-negative. Then we continue with the
preprocessing and query algorithms as described in Section
\ref{sec:dso1}. Finally, at the end of the query when we computed
$d_G(s,t,F)$ with respect to the weights $\omega_\phi(\cdot)$, we add
to it $\phi(t) - \phi(s)$ to obtain the weight of this shortest path
$P_G(s,t,F)$ with respect to the original weights $\omega(\cdot)$.

\subsubsection{Unique Shortest Paths Assumption} \label{sec:unique}
In this section we justify the unique shortest paths assumption.

For randomized algorithms, unique shortest paths can be achieved easily by using a folklore method of adding small perturbations to the edge weights, such that all the shortest paths in the resulting graph are unique w.h.p. and a shortest path in the resulting graph is also a shortest path in the original graph.

We describe a way to define unique shortest paths in a graph $H$
which fits the algorithms we presented. First, we assume that the
weights are non-negative according to the reduction described in
Section \ref{sec:negative}. Next, let $0 < \epsilon' < 1$ be a small
enough number such that $n \cdot \epsilon' < \min \{ \omega(u,v) \ |
\ (u,v) \in E \}$. Add $\epsilon'$ to the weight of all the edges of
the graph. Then we get that all the edges have positive weights, and
every shortest path in the graph after adding $\epsilon'$ is also a
shortest path in the original graph. Now, we define the unique
shortest paths $P_H^L(s,t)$ in the graph $H$ by recursion on $L \ge
0$. For $L=0$ we define $P_H^0(s,s) = <s>$ and $d_H^0(s,s) = 0$, and
for every pair of vertices $s, t \in V, s \ne t$ we define
$P_H^0(s,t) = \emptyset$ and $d_H^0(s,t) = \infty$. For the
inductive step we need to define $P_H^L(s,t)$. Let $u_1, \ldots,
u_\ell$ be all the neighbours of $s$ which minimize $\omega(s,u_i) +
d_H^{L-1}(u_i, t)$ among all the vertices $V$. Let $u_i$ be the
vertex whose index (label) is minimal among $u_1, \ldots, u_\ell$.
We define $P_H^L(s,t) = (s,u_i) \circ P_H^{L-1}(u_i, t)$ and
$d_H^L(s,t) = \omega(s,u_i) + d_H^{L-1}(u_i,t)$, such that
$P_H^{L-1}(u_i, t)$ are uniquely defined by the induction
hypothesis. For every pair of vertices $s,t \in V$ such that the
above inductive step did not define $P_H^L(s,t)$ we define
$P_H^L(s,t) = \emptyset$ and $d_H^L(s,t) = \infty$. We define
$P_H(s,t)$ as follows. Let $X = \arg \min_{0 \le L \le n} \{
d_H^L(s,t) \}$. Then we define $P_H(s,t) = P_H^X(s,t)$.

Finally, for every $s,t \in V$  we  define $P_G(s,t) = P_G(s,t)$,
$d_G(s,t) = \omega(P_G(s,t))$, and for every $0 \le L \le n$ we define
$P^L_G(s,t) = P_G^L(s,t)$, $d^L_G(s,t) = \omega(P^L_G(s,t))$. For a subset
$F \subseteq E \cup V$ we define $P^L_G(s,t,F) = P_{G\setminus
F}^L(s,t)$, $d^L_G(s,t,F) = d_{G\setminus F}^L(s,t)$, $P_G(s,t,F) =
P_{G\setminus F}(s,t)$ and $d_G(s,t,F) = d_{G\setminus F}(s,t)$. It is
not difficult to prove the following lemma.

\begin{lemma}
For every $s,t \in V, F \subseteq E \cup V, 0 \le L \le n$ the path
$P^L_G(s,t,F)$ is a shortest path among all $s$ to $t$ paths in
$G\setminus F$ that contain $L$ edges, and the path $P_G(s,t,F)$ is a
shortest path from $s$ to $t$ in $G\setminus F$. Both $P_G(s,t,F)$ and
$P^L_G(s,t,F)$ are uniquely defined, and their lengths are $d_G(s,t,F)$
and $d^L_G(s,t,F)$ respectively.
\end{lemma}

The following lemma is also not difficult to prove.

\begin{lemma}
When running Dijkstra or the dynamic programming algorithm as
described in Section \ref{sec:dynamic-programming} in the graph $G
\setminus F$, if during the execution of the algorithm instead of
considering vertices in arbitrary order we always consider vertices
in ascending order of their labels (indices) then we compute the
unique shortest paths $P^L_G(s,t,F)$.
\end{lemma}

\section{Open Questions} \label{sec:open-questions} Here are some
open questions that immediately follow our work.

\begin{itemize}
\item Weimann and Yuster \cite{WY13} presented a randomized algebraic algorithm for constructing a DSO whose runtime is subcubic and the query has subquadratic runtime supporting $f = O(\lg n / \lg \lg n)$ edges or vertices failures. Grandoni and Vassilevska Williams \cite{GW12} presented a randomized algebraic algorithm for constructing a DSO whose runtime is subcubic and the query has sublinear runtime supporting a single ($f=1$) edge failure.
    The preprocessing algorithms of both these DSOs is randomized and algebraic, and it remains an open question if there exists a DSO with subcubic deterministic preprocessing algorithm and subquadratic or sublinear deterministic query algorithm, matching their randomized equivalents?

\item Both the DSOs of Weimann and Yuster \cite{WY13} and Grandoni and Vassilevska Williams \cite{GW12} use the following randomized procedure ({\sl e.g.}, Lemma 2 in \cite{GW12}):
Let $0 < \epsilon < 1$, $1 \le f \le \epsilon \lg n/ \lg \lg n$, $L
= n^{\epsilon/f}$,  and let $C > 0$ be a
 large enough constant.
Sample $s = L^f \cdot C \log n$ graphs $\{G_1, \ldots ,G_s\}$, where
each $G_i$ is obtained from $G$ by independently removing each edge
with probability $(1/L)$. For $C$ large enough, it holds whp that
for every $(s,t,e)$ for which there exists a replacement path
$P_G(s,t,F)$ on at most $L$ nodes, there is at least one $G_i$ that
does not contain $F$ but contains at least one replacement path for
$(s,t,F)$ on at most $L$ edges. The time to compute the graphs $G_1,
\ldots, G_s$ using randomization is $\Otilde(m \cdot s) =
\Otilde(n^{2+\epsilon})$.

We ask what is the minimum $s$ such that we can deterministically
compute such graphs $G_1, \ldots, G_s$ in $\Otilde(n^2 \cdot s)$
time such that the above property holds (that for every $(s,t,e)$
for which there exists a replacement path $P_G(s,t,F)$ on at most $L$
nodes, there is at least one $G_i$ that does not contain $F$ but
contains at least one replacement path for $(s,t,F)$ on at most $L$
edges).

The randomized algorithm has a simple solution (as mentioned above,
sample $s = L^f \cdot C \log n$ graphs $\{G_1, \ldots ,G_s\}$, where
each $G_i$ is obtained from $G$ by independently removing each edge
with probability $(1/L)$). As there are $O(n^{2f+2})$ different
possible queries $(s,t,F)$, and there are at most $O(n^{2f+3})$
intervals $P_G(s,t,F)$ (which we want to maintain) containing exactly
$n^{\epsilon/f}$ vertices, it is not difficult to prove (as done in
\cite{WY13}) that for every possible query $(s,t,F)$ there exists
whp at least one graph $G_i$ which does not contain $F$ but contains
$P_G(s,t,F)$.

It is not clear how to efficiently derandomize a degenerated version
of the above construction. We can even allow some relaxations to the
above requirements. Assume there is a list ${\cal L} = \{ (s_1, t_1,
F_1), \ldots, (s_\ell, t_\ell, F_\ell) \}$ of at most $\ell =
O(n^{2+\epsilon})$ queries which are the only queries that interest
us, and assume there is a smaller set of intervals ${\cal P} = \{
P_G(s_1, t_1, F_1), \ldots, P_G(s_\ell, t_\ell, F_\ell) \}$ each
contains exactly $n^{\epsilon/f}$ edges that we want to maintain.
Then, what is the minimum $s$ (asymptotically) such that one can
construct deterministically graphs $\{G_1, \ldots ,G_s\}$ in
$\Otilde(n^2 \cdot s)$ time, such that for every $1 \le i \le \ell$
there exists at least one graph $G_i$ which does not contain $F_i$
but contains $P_G(s_i,t_i,F_i)$?

It is an open question how to achieve this goal, even for $f=1$ and
even if we allow $s$ to be greater than
$\tilde{\Omega}(n^{\epsilon})$ (to the extend that running APSP in
the graphs $G_1, \ldots, G_s$ still takes subcubic time)?
\end{itemize}

An indirect open question is to derandomize more randomized
algorithms and data-structures in closely related fields, perhaps
utilizing some of our techniques and framework.

\bibliographystyle{plain}
\bibliography{dso}

\end{document}